\documentclass[sigconf]{acmart}
\AtBeginDocument{%
  }

\setcopyright{acmlicensed}
\copyrightyear{2018}
\acmYear{2018}
\acmDOI{XXXXXXX.XXXXXXX}
\acmConference[Conference acronym 'XX]{Make sure to enter the correct
  conference title from your rights confirmation email}{June 03--05,
  2018}{Woodstock, NY}
\acmISBN{978-1-4503-XXXX-X/2018/06}

\usepackage{multirow}
\usepackage{subfigure}
\usepackage{algorithm}
\usepackage{algorithmicx}
\usepackage[noend]{algpseudocode}

\newcommand{\ours}{UniDisc}

\begin{document}

\title{Unified Data Discovery across Query Modalities and User Intents}

\author{Tingting Wang}
\affiliation{%
	\institution{The University of Queensland}
	\city{St Lucia}
	\state{QLD}
	\country{Australia}
}
\email{tingting.wang@uq.edu.au}

\author{Shixun Huang}
\affiliation{%
	\institution{The University of Wollongong}
	\city{Wollongong}
	\state{NSW}
	\country{Australia}
}
\email{shixunh@uow.edu.au}

\author{Zhifeng Bao}
\affiliation{%
	\institution{The University of Queensland}
	\city{St Lucia}
	\state{QLD}
	\country{Australia}
}
\email{zhifeng.bao@uq.edu.au}

\author{J. Shane Culpepper}
\affiliation{%
	\institution{The University of Queensland}
	\city{St Lucia}
	\state{QLD}
	\country{Australia}
}
\email{s.culpepper@uq.edu.au}

\author{Shazia Sadiq}
\affiliation{%
	\institution{The University of Queensland}
	\city{St Lucia}
	\state{QLD}
	\country{Australia}
}
\email{shazia@eecs.uq.edu.au}

\author{Volkan Dedeoglu}
\affiliation{%
	\institution{CSIRO's Data61}
	\city{Pullenvale}
	\state{QLD}
	\country{Australia}
}
\email{Volkan.Dedeoglu@data61.csiro.au}

\author{Reza Arablouei}
\affiliation{%
	\institution{CSIRO's Data61}
	\city{Pullenvale}
	\state{QLD}
	\country{Australia}
}
\email{Reza.Arablouei@data61.csiro.au}



\begin{abstract}
Data discovery -- retrieving relevant tables from a data lake in response to user queries -- is a fundamental building block for downstream analytics. In practice, data discovery must support different query modalities, including natural language (NL) statements and tables.
It must also accommodate diverse user intents, ranging from open-ended enrichment to task-driven inference for applications such as table question answering and fact verification.
However, most existing methods are designed for a single query modality or a specific user intent, limiting their generalizability.

We propose \ours{}, a \underline{Uni}fied Data \underline{Disc}overy framework that supports both NL statements and tables as queries, and generalizes across diverse user intents without relying on intent-specific representations or relevance modeling. 
At its core, \ours{} learns a common cross-modal representation learning model that produces unified representations for queries of different modalities and candidate tables, enabling uniform relevance assessment across discovery scenarios.

Learning such a common model typically requires large collections of labeled (query, table) pairs to bridge the substantial semantic gap between queries of different modalities and candidate tables, which is expensive to obtain.
To this end, we introduce a cross-modal graph learning model that learns unified representations for both queries and candidate tables under limited supervision by capturing naturally occurring contextual signals in a data lake. NL statements and tables are modeled as nodes in a heterogeneous graph with multiple edge types, over which dual-view neighbor aggregation and joint optimization are applied to learn robust and context-aware representations.
These representations then form the basis for flexible relevance estimation during retrieval.

Extensive experiments on seven datasets, covering three types of NL statements -- short texts without explicit intents, clearly stated questions, and unverified factual claims -- show that \ours{} consistently outperforms strong baselines on both NL- and table-based discovery. Moreover, \ours{} improves downstream applications such as table question answering and fact verification, while maintaining strong performance with only 1\% of the training data, demonstrating high label efficiency and scalability.

\end{abstract}

\maketitle

\section{Introduction}
\label{sec:intro}

\begin{figure}[t]
\setlength{\abovecaptionskip}{0cm}
\setlength{\belowcaptionskip}{-0.5cm}
\centering
\subfigtopskip=0pt
\subfigbottomskip=0pt
\subfigcapskip=-4pt
\subfigure[Open-ended text enrichment]{\label{subfig:news} \includegraphics[width=81mm]{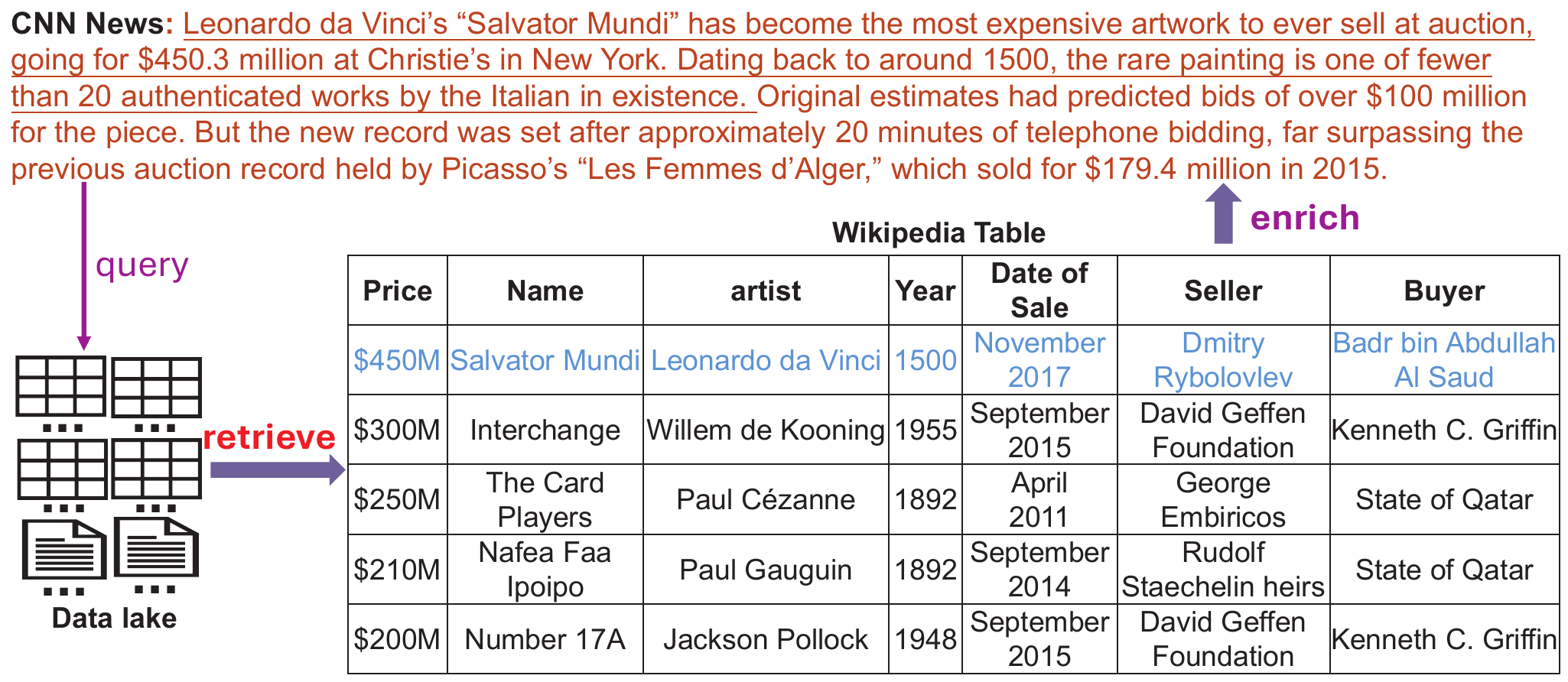}}\\
\subfigure[Open-ended table enrichment]{\label{subfig:join} \includegraphics[width=71mm]{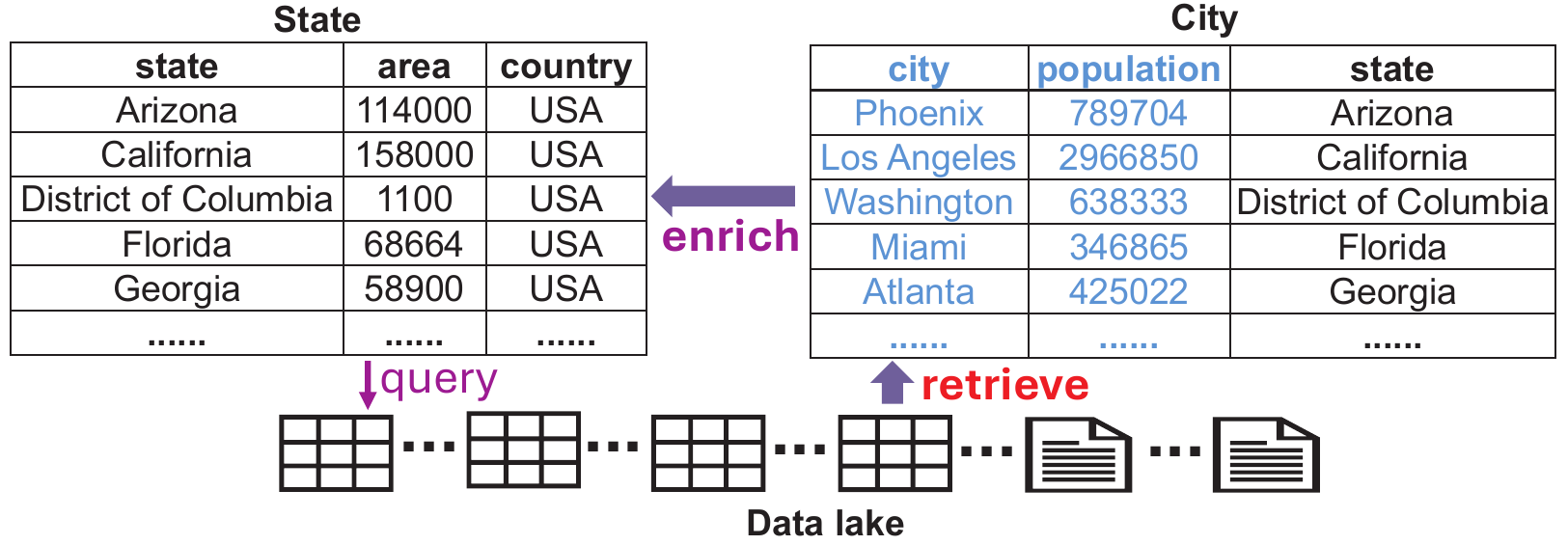}}\\
\subfigure[Table-based fact verification 
]{\label{subfig:TFV} \includegraphics[width=74mm]{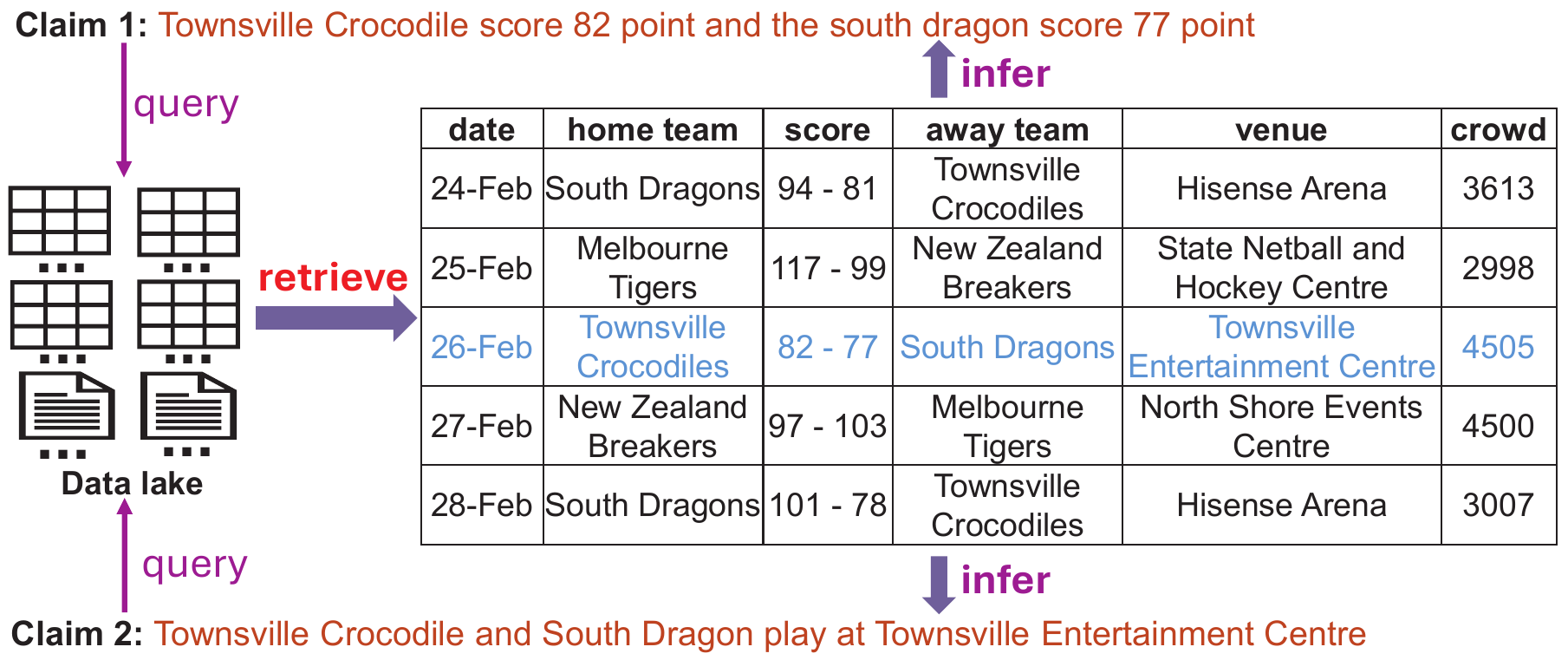}}\\
\subfigure[Table question answering 
]{\label{subfig:TQA} \includegraphics[width=81mm]{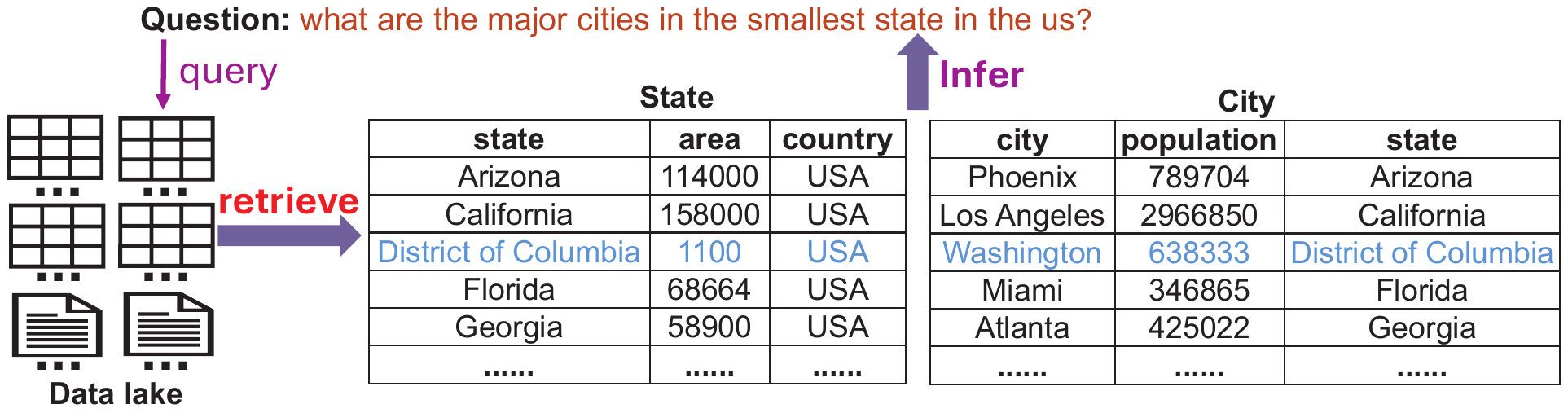}}\\
\caption{\label{fig:example} Data discovery for diverse user intents.
Retrieved tables provide the information (in blue) to fulfill user intents.
}
\end{figure}

The proliferation of data collection technologies has led to massive, often uncurated repositories known as \emph{data lakes}~\cite{NargesianZMPA19}. 
While rich in potential, data lakes pose significant challenges for locating data tailored to specific downstream tasks. 
As a result, \emph{data discovery} -- the process of identifying relevant tables from a data lake in response to user queries -- has emerged as a fundamental component of modern data science pipelines, attracting increasing attention from both academia and industry~\cite{BrickleyBN19,CasteloRSBCF21,BogatuFP020,FernandezAKYMS18}. 

Despite its fundamental role, data discovery in practice involves diverse user intents -- either explicit or implicit -- giving rise to two representative discovery scenarios, as illustrated in Fig.~\ref{fig:example}:

\noindent$\bullet$ \emph{Open-ended Enrichment.} 
In this scenario, the query does not explicitly specify a concrete information need.
Enrichment can be instantiated across different query modalities, including text enrichment~\cite{LeesBKS0021,SilvaB23} and table enrichment~\cite{Dong0NEO23,GuoMHCG25}.
For example, given a short text describing an auction event (Fig.~\ref{subfig:news}), a user may retrieve structured tables to fill in missing information, such as the buyer and seller.
Similarly, given a table containing state-level information in the U.S. (Fig.~\ref{subfig:join}), the user may augment its attributes by retrieving joinable tables (e.g., \texttt{City}).

\noindent$\bullet$ \emph{Task-driven Inference.} In this scenario, data discovery is explicitly guided by downstream tasks, with retrieved tables serving as inputs for subsequent inference.
In table-based fact verification~\cite{OpenTFV22,ChaiGZF021}, users issue factual claims and seek relevant tables as evidence to support or refute the claims.
For example, a table of sports team statistics can serve as evidence for multiple factual claims, as shown in Fig.~\ref{subfig:TFV}.
In table question answering (QA)~\cite{WangF23,balaka2025pneuma}, users issue questions, and relevant table identification is a necessary precondition for answer inference.
For example, the tables \texttt{State} and \texttt{City} can be retrieved to infer the answer to the question in Fig.~\ref{subfig:TQA}.

Although user intents and query modalities vary widely, they share a common goal: retrieving tables that are relevant to a query, where the query is a natural language (NL) statement or a table. This motivates the \emph{unified data discovery problem}: given a query in either modality, retrieve relevant tables from a data lake in an intent-agnostic manner across discovery scenarios.

However, existing methods largely fall short of this goal, as they are typically designed for specific query modalities or user intents. Most prior work supports a single query modality -- either tables~\cite{BogatuFP020,GalhotraGF23,CasteloRSBCF21} or NL statements~\cite{SilvaB23,WangF23,LeesBKS0021}. 
Within a given query modality, these methods are typically optimized for a specific user intent in a particular discovery scenario, such as open-ended enrichment (e.g., joinable table search~\cite{Dong0NEO23,GuoMHCG25,DongT0O21} or unionable table search~\cite{NargesianZPM18,KhatiwadaFSCGMR23,FanWLZM23}) or task-driven inference (e.g., table retrieval for QA~\cite{WangF23,balaka2025pneuma,wummqa}).

CMDL~\cite{EltabakhKEA23} partially addresses this limitation by learning unified representations for both NL and table queries, but still relies on intent-specific relevance modeling. Recent advances in large language models (LLMs) demonstrate strong general relevance assessment capabilities~\cite{ArabzadehC25a,RatheeMA25,TakehiVSS25}. However, applying LLMs to data discovery typically requires exhaustive query–table evaluation, resulting in prohibitive cost and latency (See our experiments in Sec.~\ref{sec:NL2table}). Consequently, extending existing methods to broader query modalities or user intents often incurs task-specific adaptations, additional supervision, or substantial computational overhead, limiting scalability and generality.

To fill this gap, we propose \ours{}, a unified data discovery framework that supports different query modalities and generalizes across diverse user intents in discovery scenarios, including open-ended enrichment and task-driven inference.
The design of \ours{} is guided by a central question: \emph{how can relevance between a query -- either an NL statement or a table -- and candidate tables be assessed without relying on intent-specific representations or relevance modeling?} 
A natural answer is to design a common cross-modal representation learning model.
Such a model embeds queries of different modalities and tables into a shared space, providing a basis for uniformly estimating relevance across discovery scenarios.
However, learning this model typically requires large collections of labeled (query, table) pairs to bridge the semantic gap between NL statements and tables.
In practice, obtaining such labeled pairs is expensive -- even synthetic generation incurs non-trivial effort~\cite{EltabakhKEA23}. 

\smallskip
\noindent\textbf{Contributions}. To address this, we make contributions as follows:

\noindent$\bullet$ We introduce a graph learning model to learn unified representations for queries of different modalities and candidate tables under \textbf{limited supervision}, which is central to \ours{}.
This model amplifies the available supervision by leveraging rich contextual signals that naturally exist in data lakes.
For example, tables with overlapping content (e.g., tables in Fig.~\ref{subfig:TQA}) are often relevant to the same query, while NL statements about related topics (e.g., claims in Fig.~\ref{subfig:TFV}) frequently reference the same table.
We encode such signals in a heterogeneous graph where NL statements and tables are nodes, and edges capture rich cross-modal and intra-modal contextual signals, including semantic or co-occurrence links (between NL statements and tables), structural or semantic similarity (between tables), and topical relatedness (between NL statements).
By propagating information over this graph, the model learns robust representations that capture relational context unavailable from isolated query-table pairs (Secs.~\ref{sec:representation}--\ref{sec:opt}).

\noindent$\bullet$ We materialize the learned representations into a vector index to enable efficient retrieval.
Built on this index, \ours{} introduces an adaptive retrieval module that operates on the learned representations to uniformly support relevance assessment across discovery scenarios, including retrieving all tables identified as relevant and returning a top-$k$ ranked list (Sec.~\ref{sec:retrieval}). 

\noindent$\bullet$ We conduct extensive experiments on seven real-world datasets against strong baselines, including CMDL, demonstrating that \ours{} achieves superior effectiveness and efficiency across diverse discovery scenarios. In particular, \ours{} outperforms baselines for NL-based discovery over short texts without explicit intents, clearly stated questions, and unverified factual claims, as well as for table-based discovery using tables as queries, achieving significant effectiveness gains and orders-of-magnitude speed-ups. Moreover, by improving retrieval quality, \ours{} yields measurable end-to-end performance gains on downstream tasks such as table-based fact verification, and remains robust under limited supervision (as little as 1\% training data) and sparse graph structures (Sec.~\ref{sec:experiment}).

\section{A Unified Data Discovery Framework}
\label{sec:overview}

In this section, we first present the unified data discovery problem setting and then introduce the design of \ours{}, a unified data discovery framework tailored to this setting.

\begin{figure*}[t]
\setlength{\abovecaptionskip}{0cm}
\setlength{\belowcaptionskip}{-0.5cm}
\centering
\includegraphics[width=180mm]{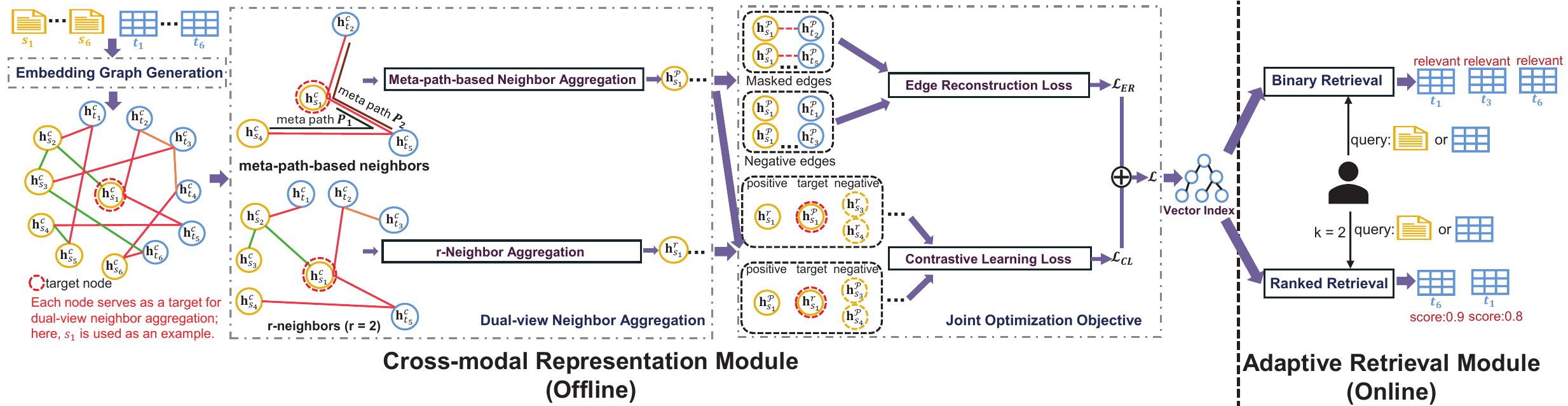}
\caption{\label{fig:framework} Overview of \ours{}.
}
\end{figure*}

\vspace{-0.8em}
\subsection{Problem Setting}
\label{sec:problem}

We study the \emph{unified data discovery problem} in data lakes, where users seek to identify tables relevant to a given query in an intent-agnostic
manner.
As illustrated in Fig.~\ref{fig:example}, a query may take different modalities, including a natural language (NL) statement or a table, and may reflect diverse user intents.

User intents may be \emph{explicit}, targeting task-driven inference -- such as a clearly stated question (Fig.~\ref{subfig:TQA}) or an unverified factual claim (Fig.~\ref{subfig:TFV}) -- or \emph{implicit}, targeting open-ended enrichment, such as a short text (Fig.~\ref{subfig:news}) or a table (Fig.~\ref{subfig:join}) that reflects an underlying information need.
\ours{} treats both NL statements and tables as valid query modalities and aims to support data discovery uniformly across these modalities.

A data lake comprises a collection of heterogeneous tables together with a set of NL statements, which may arise from user queries or downstream applications.
Formally, we define a data lake as $\mathcal{D} = (T, S)$, where $T = \{t_1, \ldots, t_{|T|}\}$ denotes the collection of tables and $S = \{s_1, \ldots, s_{|S|}\}$ denotes the collection of NL statements.

\begin{definition}[Unified Data Discovery]
\label{def:problem}
Given a data lake $\mathcal{D} = (T, S)$ and a query $q$, which can be either an NL statement or a table, 
it aims to retrieve a set of relevant tables $T_{\mathrm{rel}} \subseteq T$ based on the relevance between the query $q$ and candidate tables in $T$.
Two common relevance formulations are considered:

\noindent(1) A relevance scoring function $r: q \times t \rightarrow \mathbb{R}$, where $r(q, t)$ measures the relevance score between $q$ and a candidate table $t$.
Based on this, $T_{\mathrm{rel}}$ is returned as a top-$k$ ranked list of tables, i.e.,
$
T_{\mathrm{rel}} = \langle t_1, \ldots, t_k \rangle$
such that
$r(q,t_1) \ge \cdots \ge r(q,t_k),
\; t_i \in T.
$

\noindent(2) A binary relevance function $f: q \times t \to \{0, 1\}$, where $f(q, t) = 1$ indicates that a candidate table $t$ is relevant to $q$, and $f(q, t) = 0$ otherwise. Based on this, $T_{\mathrm{rel}}$ is returned as all tables identified as relevant, i.e.,
$
T_{\mathrm{rel}} = \{t \in T \mid f(q,t)=1\}.
$
\end{definition}


\vspace{-1em}
\subsection{Framework Overview}
\label{sec:framework}

To solve the unified data discovery problem, we propose \ours{}, a \underline{Uni}fied Data \underline{Disc}overy framework that integrates offline representation learning with online retrieval.
As illustrated in Fig.~\ref{fig:framework}, \ours{} consists of two modules with clearly separated responsibilities.
Together, these two modules enable \ours{} to support queries of different modalities and diverse discovery scenarios.

\subsubsection{Cross-modal Representation Module}
\label{sec:cross_model}

This module serves as the core of \ours{}, aiming to learn unified representations for queries of different modalities and tables under limited supervision.
To this end, it exploits contextual signals that exist in data lakes by modeling both NL statements and tables as nodes in a heterogeneous graph.
In this graph, edges encode multiple types of contextual relationships, including \emph{cross-modal associations} (e.g., claims supported by an evidence table in Fig.~\ref{subfig:TFV} or tables required to answer a question in Fig.~\ref{subfig:TQA}) and \emph{intra-modal relationships} (e.g., joinable tables such as \texttt{State} and \texttt{City} in Fig.~\ref{subfig:join}, or topically related claims in Fig.~\ref{subfig:TFV}).
Building on this graph, the learning process integrates two complementary aggregation mechanisms to capture both semantic and structural context, and jointly optimizes two training objectives to align representations across modalities while preserving relational consistency within each modality.
As a result, the learned representations reside in a shared embedding space that supports direct relevance estimation between queries and tables.
These embeddings are materialized and indexed, serving as the foundation for efficient retrieval across diverse discovery scenarios. We describe the details of this module in Secs.~\ref{sec:representation}--\ref{sec:opt}.

\subsubsection{Adaptive Retrieval Module}
\label{sec:retrieval}
Given the unified representations learned by the cross-modal representation module, this module performs data discovery by retrieving relevant tables for a given query based on the constructed index efficiently and flexibly.
At query time, a query $q$ -- either an NL statement or a table -- is mapped to its embedding $\mathbf{h}_q$ using the trained representation model, which is then used to estimate relevance with the embeddings $\mathbf{h}_t$ of candidate tables $t \in T$.
Following Definition~\ref{def:problem}, the module supports retrieval under two relevance formulations, which we refer to as \emph{ranked retrieval} and \emph{binary retrieval}.

\smallskip
\noindent\textbf{Ranked Retrieval.}
In ranked retrieval, data discovery proceeds via relevance ranking, where the user specifies the number of tables to retrieve, as commonly required in scenarios such as table question answering for focused reasoning.
The module computes a relevance score $r$ using cosine similarity.
Given the query embedding $\mathbf{h}_q$, tables are retrieved from the constructed index, ranked by $r(\mathbf{h}_q, \mathbf{h}_t)$, and the top-$k$ results are returned.
This retrieval procedure runs in sublinear time with respect to the number of tables $|T|$, typically $O(\log |T|)$ per query in practice, and avoids exhaustive similarity computation over the data lake.

\smallskip
\noindent\textbf{Binary Retrieval.}
In binary retrieval, relevance is determined by a binary function $f$, and all tables predicted as relevant are returned. This setting fits discovery scenarios where users do not specify the number of tables to retrieve and aim to obtain comprehensive results. For example, users may want to obtain tables that together provide comprehensive evidence in table-based fact verification.
To efficiently support this setting, the module adopts a two-stage retrieval strategy.
First, it uses the constructed index to retrieve a candidate set of tables for the query embedding $\mathbf{h}_q$, where the candidate size is chosen to be sufficiently larger than typical result sets to mitigate potential recall loss.
Second, a lightweight multi-layer perceptron (MLP) classifier is applied to estimate the binary relevance function $f$.
For each query-table pair $(\mathbf{h}_q, \mathbf{h}_t)$ within the candidate set, $f(\mathbf{h}_q, \mathbf{h}_t)$ determines relevance, and the subset predicted as relevant is returned.
This design avoids evaluating all tables in the data lake while naturally supporting variable-size discovery results.
The overall inference time consists of sublinear candidate retrieval followed by linear-time relevance prediction over the candidate set.
In practice, the candidate size is bounded independently of $|T|$, ensuring scalability to large-scale data lakes.

\section{Graph-Based Cross-Modal Representation Learning}
\label{sec:representation}

As discussed in Sec.~\ref{sec:cross_model}, the cross-modal representation module of \ours{} operates offline to learn shared representations across data modalities in the data lake, providing the foundation for \ours{}.
A key modeling challenge is to enable direct comparison between unstructured NL statements and structured tables without relying heavily on labeled (NL statement, table) pairs.
Prior approaches address this by embedding both modalities into a shared space~\cite{EltabakhKEA23,LeesBKS0021,SilvaB23,WangF23}, bringing semantically aligned pairs closer together.
However, such methods typically require substantial labeled supervision or rely on pseudo-pair synthesis~\cite{EltabakhKEA23}, incurring additional computational overhead.

\smallskip
\noindent\textbf{Key Observation.}
Beyond direct supervision from labeled (NL statement, table) pairs, data lakes naturally exhibit rich \emph{contextual signals} that can serve as indirect supervision.
For instance, tables with similar schemas or overlapping rows often correspond to related queries, while NL statements within a similar topic frequently reference similar table content.
These semantic and relational dependencies provide valuable signals for representation learning, yet are often underexploited in existing cross-modal methods.

\smallskip
\noindent\textbf{Graph-based Cross-Modal Representation Learning.}
Motivated by this observation, we instantiate the cross-modal representation module using a graph learning model.
We model NL statements and tables as nodes in a heterogeneous graph, where edges encode both cross-modal and intra-modal contextual relationships.
Each node iteratively refines its representation by interacting with different types of neighbors, allowing information to propagate across modalities and capture contextual signals.

Through this graph modeling, representations of NL statements and tables are enriched by semantically related or co-occurring nodes of either type.
As a result, the model learns high-quality cross-modal representations without relying heavily on labeled (NL statement, table) pairs.
As illustrated in Fig.~\ref{fig:framework}, the cross-modal representation module consists of three key components, which we describe in the following sections:

\noindent(1) \emph{Embedding Graph Generation.} To support cross-modal representation learning under limited supervision, we construct a heterogeneous graph with NL statements and tables as nodes, connected by three types of edges: NL statement–Table, Table–Table, and NL statement–NL statement.
Rather than encoding node contents independently, we project heterogeneous node contents into a shared latent space by modeling inter-content dependencies using a bi-directional LSTM, which captures contextual relationships among tabular and textual elements.
This allows heterogeneous nodes to be comparable while retaining fine-grained semantic cues (Sec.~\ref{sec:nodecontent}).

\noindent(2) \emph{Dual-View Neighbor Aggregation.}
To effectively exploit contextual signals in the graph, we adopt a dual-view aggregation strategy.
Meta-path-based aggregation captures semantically meaningful relational patterns from typed interactions between NL statements and tables, while $r$-neighbor aggregation incorporates broader structural context by hierarchically aggregating information from higher-order neighbors.
These two views provide complementary signals for refining node representations (Sec.~\ref{sec:neighbor}).

\noindent(3) \emph{Joint Optimization Objective.} 
To learn robust node representations under available supervision, we jointly optimize a contrastive loss and an edge reconstruction loss.
The former enforces consistency across representations learned from different aggregation views, while the latter preserves observed relational structure among nodes.
Together, these objectives encourage view-consistent and structurally faithful representations (Sec.~\ref{sec:opt}).

\vspace{-0.6em}
\section{Embedding Graph Generation}
\label{sec:nodecontent}

To support cross-modal representation learning under limited supervision, \ours{} constructs a heterogeneous graph over tables and NL statements to capture contextual relationships, and then encodes node contents into a unified embedding space, forming a heterogeneous embedding graph.

\smallskip
\noindent\textbf{Heterogeneous Graph.}
We construct a heterogeneous graph $G = (V, E)$ based on the problem setting in Sec.~\ref{sec:problem}.
The node set $V = S \cup T$ consists of two types of nodes: \emph{NL statements} ($s \in S$) and \emph{tables} ($t \in T$).
The edge set $E$ contains three types of edges:

\noindent$\bullet$ \emph{NL statement-Table edges} encode co-occurrence or semantic relevance relationships between NL statements and tables, capturing diverse NL-based discovery intents.
During training, these edges are initialized with limited ground-truth relevance annotations, enabling the model to infer relevance for unseen pairs.

\noindent$\bullet$ \emph{Table-Table edges} encode semantic or structural similarity between tables.
Their construction is adapted to downstream requirements.
When ground-truth links are available (e.g., in table-based discovery), we use them directly; otherwise, we estimate table relevance scores and connect each table to its top-$K$ most similar neighbors.

\noindent$\bullet$ \emph{NL statement-NL statement edges} encode topical relatedness among NL statements.
We construct these edges by computing pairwise similarity scores and connecting each NL statement to its top-$K$ most similar peers.

These edges provide complementary cross-modal and intra-modal contextual signals, which are crucial for representation learning under sparse supervision. 

\smallskip
\noindent\emph{Remarks}.
In our experiments, we employ CMDL~\cite{EltabakhKEA23} to estimate table-table relevance, as it leverages both syntactic and semantic signals corresponding to the relationships we aim to encode.
For NL statement-NL statement edges, we use BM25~\cite{RobertsonZ09} due to its strong empirical performance~\cite{Kong0SSLFRK24}.
We emphasize that these methods are concrete instantiations of the underlying relationship types, and our framework does not rely on any particular choice.

\smallskip
\noindent\textbf{Unified Node Content Encoding.}
A key challenge in this heterogeneous graph is to unify heterogeneous node contents across modalities.
NL statements and tables differ substantially in structure, semantic granularity, and raw feature dimensionality, which makes direct comparison difficult.
Therefore, a modality-agnostic encoding is required to project heterogeneous node contents into a shared latent space before graph-based aggregation.

To this end, we associate each node $v$ with a content set $C_v = \{\mathbf{x}^v_1, \ldots, \mathbf{x}^v_{|C_v|}\}$, where each content element $\mathbf{x}^v_i \in \mathbb{R}^{d_c}$ is a feature vector.
For an NL statement $s \in S$, $C_s$ contains a single textual feature.
For a table $t \in T$, $C_t$ consists of features associated with all columns and the table caption, if available.


A simple baseline is to use a linear transformation for each node type to map the content vectors into a shared space. 
Specifically, for a node $v$ of type $a$ (i.e., NL statement or table), we have:
\vspace{-0.5em}
\begin{equation}
\label{eq:linear}
\mathbf{h}_{v}^c = \textstyle \sum_{i = 1}^{|C_v|}\mathbf{W}_a\mathbf{x}^{v}_i
\vspace{-0.6em}
\end{equation}
where $\mathbf{h}_{v}^c \in \mathbb{R}^{d \times 1}$, $d$ is the content embedding dimension, $\mathbf{W}_a \in \mathbb{R}^{d \times d_c}$ is a trainable projection matrix specific to node type $a$.

However, a simple linear approach fails to capture dependencies between content elements. 
This is particularly important for tables, where contextual relationships between columns (e.g., attribute correlation) can be crucial to the interpretation and reasoning needed in tasks such as answering questions or verifying claims.
To overcome this limitation, we employ a bi-directional LSTM (Bi-LSTM) encoder~\cite{ZhangSHSC19} that models inter-content dependencies and context. 
Specifically, the content embedding for a node $v$ is computed as:
\vspace{-0.3em}
\begin{equation}
\label{eq:lstm}
\mathbf{h}_{v}^c = {{\textstyle \sum_{i = 1}^{|C_v|}[\overrightarrow{\mathrm{LSTM}}(\mathbf{x}^v_1, ..., \mathbf{x}^v_{|C_v|})_i \oplus \overleftarrow{\mathrm{LSTM}}(\mathbf{x}^v_1, ..., \mathbf{x}^v_{|C_v|})_i]} / {|C_v|}}
\vspace{-0.3em}
\end{equation}
where $\overrightarrow{\mathrm{LSTM}}(\cdot)_i$ and $\overleftarrow{\mathrm{LSTM}}(\cdot)_i$ denotes the forward and backward LSTM hidden states at the $i$-th position in the content set, respectively.
The operator $\oplus$ denotes concatenation. 
This yields a fixed-length embedding $\mathbf{h}_v^c \in \mathbb{R}^{d \times 1}$ that incorporates contextual information from both directions of the content set. 
The final embedding is obtained by averaging the forward and backward hidden states across all content elements. 
This yields a heterogeneous embedding graph with unified node content representations, which is used for subsequent graph-based aggregation.

\begin{example}
As illustrated in Fig.~\ref{fig:nodecontent}, an example of the constructed heterogeneous graph contains six NL statement nodes and six table nodes. 
The red lines denote NL statement-Table edges, green lines represent NL statement-NL statement edges, and orange lines represent the 
Table–Table edges.

We initialize each content element of an NL statement or a table using a pretrained BERT model~\cite{DevlinCLT19}, and then apply a Bi-LSTM to encode node contents. 
For each NL statement $s_i$ ($i \in \{1, 2, 3, 4, 5, 6\}$), the output of the Bi-LSTM is used as the content embedding $\mathbf{h}_{s_i}^c$. 
For each table $t_i$ ($i \in \{1, 2, 3, 4, 5, 6\}$), each content element is processed by the Bi-LSTM sequentially. 
The forward and backward hidden states at each position are concatenated together, and the resulting vectors are aggregated using mean pooling to obtain a table content embedding $\mathbf{h}_{t_i}^c$.
This process results in a heterogeneous embedding graph with unified node content representations.
\end{example}

\vspace{-0.5em}
\section{Dual-view Neighbor Aggregation}
\label{sec:neighbor}

Given the heterogeneous embedding graph constructed in the previous section, the next challenge is how to aggregate contextual signals from neighboring nodes to refine node representations. A key design consideration is the choice of neighbors for aggregation, as restricting aggregation to only direct (first-order) neighbors may exclude informative contextual signals from distant nodes. 

To address this, we employ two complementary neighbor aggregation views: \emph{meta-path-based neighbor aggregation} and \emph{$r$-neighbor aggregation}. The former captures semantically meaningful relational patterns arising from heterogeneous interactions between tables and NL statements, while the latter incorporates broader structural context by hierarchically aggregating information from higher-order neighbors. Together, these two views provide complementary signals for refining cross-modal node representations.


\begin{figure}[t]
\setlength{\abovecaptionskip}{0cm}
\setlength{\belowcaptionskip}{-0.6cm}
\centering
\includegraphics[width=60mm]{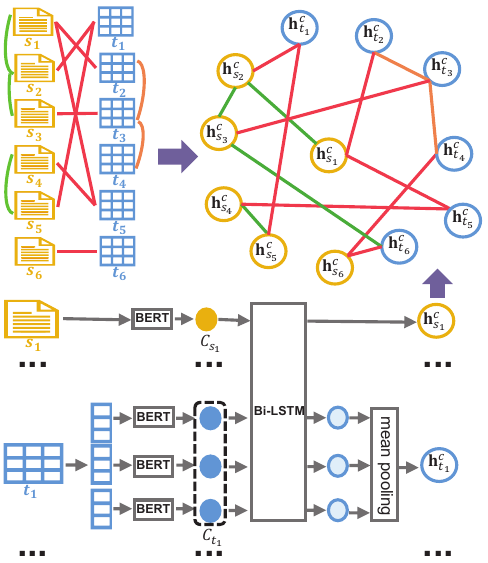}
\caption{\label{fig:nodecontent}{Embedding Graph Generation.} 
}
\end{figure}

\vspace{-0.5em}
\subsection{Meta-path-based Neighbor Aggregation}
\label{sec:mpnei}

Meta-path-based neighbor aggregation captures high-order semantic relationships by modeling typed relational patterns in a heterogeneous graph.
A meta-path defines how two types of nodes are indirectly connected via a sequence of node and edge types~\cite{fu2020magnn}, enabling aggregation beyond direct connections.

We define two meta-paths commonly arising in discovery scenarios.
The first is \emph{NL statement-Table-NL statement} ($s-t-s$), which connects NL statements that reference the same table.
This pattern frequently appears in table-based fact verification~\cite{GuF0NZ022,ChenWCZWLZW20}, where multiple claims are supported or refuted by the same evidence table.
The second is \emph{Table-NL statement-Table} ($t-s-t$), which connects tables referenced by the same NL statement.
This pattern is common in table question answering~\cite{PalYKR23,wummqa}, where multiple tables are required to answer a question.

For a given meta-path $P$, we define the \emph{meta-path-based neighbors} $\mathcal{N}^P(v)$ of a node $v$ as the set of nodes reachable from $v$ via any instance of $P$.
When multiple path instances connect the same node pair, each instance is treated separately during aggregation to account for different relational contexts.
For symmetric meta-paths, $\mathcal{N}^P(v)$ also includes a self-reference to $v$.



The design of meta-path-based neighbor aggregation considers three key goals:
(i) preserving information from intermediate nodes along a meta-path to capture multi-hop dependencies,
(ii) integrating multiple meta-paths to uncover diverse and complementary semantic relationships, and
(iii) differentiating neighbor contributions based on both the meta-path and contextual relationships.

To realize these design goals, we employ a two-level attention mechanism that combines \emph{node-level} and \emph{path-level} attention.
Node-level attention differentiates the contributions of neighbors along a given meta-path, including both end and intermediate nodes, while path-level attention learns the relative importance of different meta-paths when integrating semantic signals.
These two mechanisms enable weighted aggregation across and within meta-paths, capturing fine-grained semantic relationships while accounting for broader contextual signals.






\smallskip
\noindent\textbf{Node-level attention.} To compute the relative weight contribution for each neighbor from a meta-path to a target node, we employ a self-attention mechanism~\cite{VaswaniSPUJGKP17}.
Therefore, for a node pair $(v, v')$ connected by a meta-path $P$, the attention score $\alpha_{vv'}^P$ represents the influence of neighbor node $v'$ on the target node $v$. 
It is computed in a normalized form as follows:
\vspace{-0.3em}
\begin{equation}
\alpha^P_{v,v'} = {{\exp\{\mathrm{LeakyReLU}(\mathbf{a}_P^T[\mathbf{h}_v^c \parallel \mathbf{h}_{v'}^c])\}} \over {\sum_{v'' \in \mathcal{N}^P(v)}\exp\{\mathrm{LeakyReLU}(\mathbf{a}_P^T[\mathbf{h}_v^c \parallel \mathbf{h}_{v''}^c])\}}},
\vspace{-0.3em}
\end{equation}
where $\mathrm{LeakyReLU}$ denotes the leaky rectified linear unit function, $\mathbf{a}_P \in \mathbb{R}^{2d \times 1}$ is the learnable attention vector specific to a meta-path $P$, and $\parallel$ denotes the vector concatenation. 
This node-level attention mechanism enables differentiated weighting of neighbors along a given meta-path, capturing fine-grained contextual dependencies.

By aggregating the content embeddings from the meta-path-based neighbors using the weights derived from the corresponding attention scores, the aggregated neighbor embedding for a target node $v$ based on meta-path $P$ is computed as follows:
\vspace{-0.3em}
\begin{equation}
\mathbf{h}_v^P = \sigma\left(\textstyle \sum_{v' \in \mathcal{N}^P(v)}\alpha^P_{v,v'}\mathbf{h}_c(v')\right) 
\vspace{-0.3em}
\end{equation}
where $\mathbf{h}_v^P \in \mathbb{R}^{d \times 1}$, $\sigma$ is a nonlinear activation function, $\mathbf{h}_c(v')$ is the content embedding of neighbor $v'$, and $\alpha^P_{v, v'}$ denotes the attention-based weight assigned to node $v'$ relative to $v$ under path $P$. 

\smallskip
\noindent\textbf{Path-level attention.} After aggregating the node information from each meta-path, the node embedding can be further refined by integrating the semantic signals revealed from other meta-paths.
Following prior work on heterogeneous graph modeling~\cite{fu2020magnn}, we apply a path-level attention mechanism to learn the relative importance of different meta-paths, which we then use to aggregate the path-specific vector representations as follows:
\begin{equation}
\label{eq:MPemb}
\mathbf{h}^{\mathcal{P}}_v = \sigma\left(\textstyle \sum_{P \in \mathcal{P}}\beta^P \mathbf{h}_v^P\right)
\end{equation}
where $\mathbf{h}^{\mathcal{P}}_v \in \mathbb{R}^{d \times 1}$  is the path-aggregated embedding, 
$\mathcal{P}$ denotes the path set in a heterogeneous graph, and $\beta^{P}$ is the attention-based weight assigned to the contribution of path $P$ to node $v$.

To determine $\beta^P$, we first compute the path importance score $\omega^P$ by averaging the transformed path-specific vectors as
\vspace{-0.3em}
$$
\omega^P = {1 \over {|V|}}\textstyle \sum_{v \in V} \mathbf{q}_0^T \cdot \tanh (\mathbf{W}_0 \cdot \mathbf{h}_v^P + \mathbf{b}_0)
$$
where $\mathbf{q}_0 \in \mathbb{R}^{d \times 1}$ is the path-level attention vector, and $\mathbf{W}_0 \in \mathbb{R}^{d \times d}$ and $\mathbf{b}_0 \in \mathbb{R}^{d \times 1}$ are shared learnable parameters. 
We then apply a softmax function to normalize the path importance scores and obtain a normalized path attention weight as follows:
\vspace{-0.3em}
$$
\beta^P = {{\exp(\omega^P)} / {\textstyle \sum_{P' \in \mathcal{P}}\exp(\omega^{P'})}}
$$

This aggregation produces a meta-path-aware embedding for each target node by combining information from multiple meta-paths using Eq.~\ref{eq:MPemb}.
We denote the resulting embedding for an NL statement node $s \in S$ as $\mathbf{h}^{\mathcal{P}}_s$ and for a table node $t \in T$ as $\mathbf{h}^{\mathcal{P}}_t$. 

\begin{figure}[t]
\setlength{\abovecaptionskip}{0cm}
\setlength{\belowcaptionskip}{-0.5cm}
\centering
\includegraphics[width=78mm]{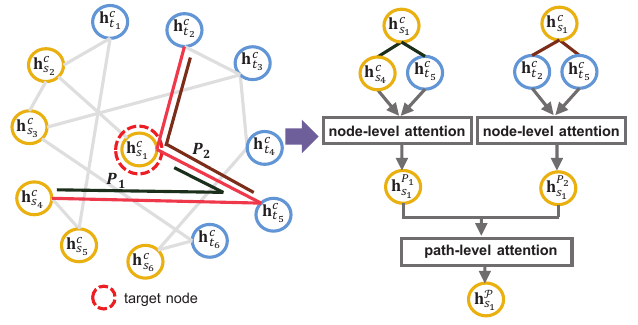}
\caption{\label{fig:mp}{Example of meta-path-based neighbor aggregation.} 
}
\end{figure}

\begin{example}
As shown in Fig.~\ref{fig:mp}, given a target node $s_1$, the meta-path instance $s_1$–$t_5$–$s_4$ for $P_1$ yields the meta-path–based neighbors $t_5$ and $s_4$, while the instance $t_2$–$s_1$–$t_5$ for $P_2$ yields the neighbors $t_2$ and $t_5$. 
We apply node-level attention to aggregate the meta-path–based neighbors for each path, resulting in the embeddings $\mathbf{h}_{s_1}^{P_1}$ and $\mathbf{h}_{s_1}^{P_2}$.
These are then combined using path-level attention to compute the final path-aggregated embedding $\mathbf{h}_{s_1}^{\mathcal{P}}$.
\end{example}

\vspace{-0.3em}
\subsection{$r$-Neighbor Aggregation}
\label{sec:rnei}

In addition to meta-path-based aggregation, we incorporate a complementary structural view by hierarchically aggregating information from higher-order neighbors.
Specifically, the 1-hop neighbors provide 0th-layer embeddings, 2-hop neighbors provide 1st-layer embeddings, and so on, up to $r$-hop neighbors for the $(r-1)$th layer embeddings. 
We refer to this hierarchical aggregation process as $r$-\emph{neighbor aggregation}.

To implement $r$-neighbor aggregation, we employ a multi-layer graph convolutional network (GCN)~\cite{KipfW17} as a generic mechanism for hierarchical neighbor aggregation.
While standard GCNs are designed for homogeneous graphs, we adapt the propagation to our heterogeneous setting by constructing a unified $(|S|+|T|) \times (|S|+|T|)$ adjacency matrix $A$ that integrates all node and edge types while preserving the symmetry required by GCN-based propagation.



We index the adjacency matrix such that NL statement-NL statement, NL statement-Table, and Table-Table edges are all encoded within a single matrix.
The weight assigned to an NL statement-Table edge $({s_i},{t_j})$ is represented as $a_{i\{j+|S|\}}$, an NL statement-NL statement edge $({s_i}, {s_j})$ as $a_{ij}$, a Table-NL statement edge $({t_i}, {s_j})$ as $a_{\{i+|S|\}j}$, and a Table-Table edge $({t_i}, {t_j})$ as $a_{\{i+|S|\}\{j+|S|\}}$.
To ensure numerical stability and incorporate self-loops, we apply standard normalization to obtain $\hat{A} = \tilde{D}^{-1/2}\tilde{A}\tilde{D}^{-1/2}$, where $\tilde{A} = A + I_{|S|+|T|}$, $I_{|S|+|T|}$ is an identity matrix, and $\tilde{D}$ is a degree matrix containing the diagonal entries $\tilde{D}_{ii} = \sum_j\tilde{A}_{ij}$.

Let $\mathbf{H}^l \in \mathbb{R}^{(|S|+|T|) \times d}$ denote the latent feature matrix at layer $l$.
The propagation rule for $r$-neighbor aggregation is defined as:
\vspace{-0.3em}
\begin{equation}
\label{eq:rnemb}
\mathbf{H}^{l+1} = \sigma(\hat{A}\mathbf{H}^{l}\mathbf{W}^l)
\vspace{-0.3em}
\end{equation}
where $\mathbf{W}^l \in \mathbb{R}^{d \times d}$ is a learnable weight matrix specific to layer $l$, and $\sigma$ is an activation function.
After $r$ layers, the final node embeddings are obtained from $\mathbf{H}^{r}$.
For an NL statement node ${s_i} \in S$, the embedding is $\mathbf{h}_{s_i}^{r} = \mathbf{H}^{r}[i]$, and for a table node ${t_i} \in T$, it is $\mathbf{h}_{t_i}^{r} = \mathbf{H}^{r}[i+|S|]$.
This produces $r$-neighbor-aware node embeddings that capture broader structural context, forming the structural aggregation view in our dual-view framework.

\begin{figure}[t]
\setlength{\abovecaptionskip}{0cm}
\setlength{\belowcaptionskip}{-0.5cm}
\centering
\includegraphics[width=82mm]{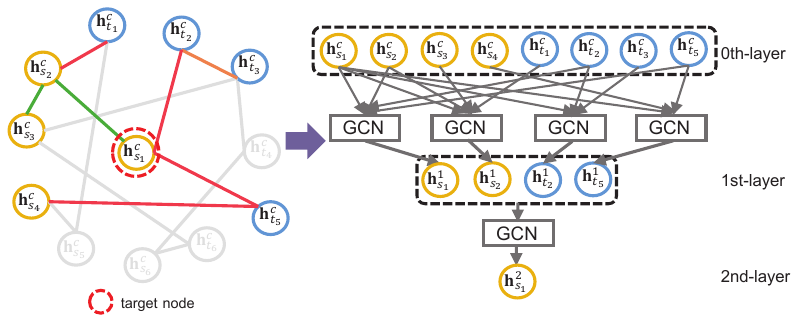}
\caption{\label{fig:rn}{Example of $r$-neighbor aggregation.} 
}
\end{figure}

\begin{example}
As shown in Fig.~\ref{fig:rn}, we perform 2-hop neighbor aggregation. 
The 1-hop neighbors for the target node $s_1$ include the NL statement $s_2$ and the tables $t_2$ and $t_5$. 
By aggregating the content embeddings (i.e., the 0-th layer embeddings) for these neighbors using a GCN, we obtain the 1st-layer embedding of $s_1$. 
Likewise, we compute the 1st-layer embeddings of $s_2$, $t_2$, and $t_5$ by aggregating the respective 1-hop neighbors. 
A second round of GCN aggregation for these 1st-layer embeddings produces the 2nd-layer embedding of $s_1$, denoted as the node embedding under 2-hop neighbor aggregation, $\mathbf{h}^2_{s_1}$.
\end{example}

\section{Joint Optimization Objective}
\label{sec:opt}

After applying the dual-view neighbor aggregation in Sec.~\ref{sec:neighbor}, each node is associated with two complementary embeddings that capture semantic and structural context, respectively.
To learn robust cross-modal representations under limited supervision, these two views need to be aligned while preserving meaningful relationships across different node types.
To this end, we adopt a joint optimization objective that combines a cross-view contrastive learning loss, which aligns the two aggregation views of the same node, with an edge reconstruction loss that explicitly models interactions between tables and NL statements.


\subsection{Contrastive Learning Loss}
For each node $v \in V$, we obtain two embeddings: $\mathbf{h}_{v}^{\mathcal{P}}$ from meta-path-based aggregation and $\mathbf{h}_{v}^{r}$ from $r$-neighbor aggregation.
These embeddings capture complementary semantic and structural contexts and are treated as two views of the same node. To align the two views, we adopt a cross-view contrastive learning objective.
For a given node $v$, we treat $\mathbf{h}_{v}^{r}$ as the positive sample for $\mathbf{h}_{v}^{\mathcal{P}}$, while embeddings of other nodes of the same type are treated as negative samples. 
Restricting negatives to the same node type avoids trivial alignment across heterogeneous node types and encourages meaningful view consistency. 
For example, in Fig.~\ref{fig:framework}, for the target node $s_1$, $\mathbf{h}_{s_1}^{r}$ serves as the positive sample, whereas $\mathbf{h}_{s_3}^{r}$ and $\mathbf{h}_{s_4}^{r}$ are used as negative samples.

We instantiate the contrastive objective using an InfoNCE-based loss~\cite{qiu2020gcc,you2020graph,wang2021self}, defined as:
$$
\mathcal{L}_{\mathrm{CL}}^{\mathcal{P}} = \sum_{v \in V}-\log{{\exp(\mathrm{sim}(\mathbf{h}_{v}^{\mathcal{P}}, \mathbf{h}_{v}^{r})/\tau)} \over {\sum_{v' \in V^-} \exp(\mathrm{sim}(\mathbf{h}_{v}^{\mathcal{P}}, \mathbf{h}_{v'}^{r})/\tau)}}
$$
where $V^-$ is the set of sampled negative nodes of the same type as $v$, $\mathrm{sim}(\cdot)$ denotes cosine similarity, and $\tau$ is a temperature parameter that adjusts the importance of hard negatives.


To facilitate loss normalization and stabilize training, we characterize an upper bound of the contrastive learning loss.
\vspace{-0.5em}
\begin{theorem}
\label{theo:loss_cl}
The contrastive learning loss $\mathcal{L}_{\mathrm{CL}}^{\mathcal{P}}$ is upper-bounded as:
$\mathcal{L}_{\mathrm{CL}}^{\mathcal{P}} \leq \sum_{v \in V} \log(|V^-| + 1)$.
\end{theorem}
\vspace{-0.5em}
\begin{proof}
The bound follows from a degenerate case where the model fails to distinguish the positive sample from negative ones.
That is, when the similarity between positive and negative pairs becomes indistinguishable. We have:
$$
\exp(\mathrm{sim}(\mathbf{h}_{v}^{\mathcal{P}}, \mathbf{h}_{v}^{r})/\tau) \approx \exp(\mathrm{sim}(\mathbf{h}_{v}^{\mathcal{P}}, \mathbf{h}_{v'}^{r})/\tau), \forall v' \in V^-
$$
Thus, the denominator in the contrastive learning loss becomes
\vspace{-0.2em}
$$
\sum_{v' \in V^-}\exp(\mathrm{sim}(\mathbf{h}_{v}^{\mathcal{P}}, \mathbf{h}_{v'}^{r})/\tau) \approx (|V^-| + 1)\exp(\mathrm{sim}(\mathbf{h}_{v}^{\mathcal{P}}, \mathbf{h}_{v'}^{r})/\tau).
$$
\vspace{-0.2em}
Hence, the loss can be approximated as
$$
\mathcal{L}_{\mathrm{CL}}^{\mathcal{P}} \approx  \sum_{v \in V} -\log\frac{1}{|V^-| + 1} = \sum_{v \in V} \log(|V^-| + 1),
$$
which establishes the stated upper bound.
\end{proof}
\vspace{-0.5em}

Based on the above upper bound, we normalize $\mathcal{L}_{\mathrm{CL}}^{\mathcal{P}}$ as
$
\hat{\mathcal{L}}_{\mathrm{CL}}^{\mathcal{P}} = \frac{\mathcal{L}_{\mathrm{CL}}^{\mathcal{P}}}{\sum_{v \in V} \log(|V^-| + 1)}
$.
Similarly, $\hat{\mathcal{L}}_{\mathrm{CL}}^{r}$ is defined by swapping the roles of the two views, treating $\mathbf{h}_{v}^{r}$ as the anchor and $\mathbf{h}_{v}^{\mathcal{P}}$ as the positive, with the meta-path-aware embeddings of other nodes as negatives.

The overall contrastive learning loss is a weighted combination of both contrastive views:
\begin{equation}
\label{eq:clloss}
\mathcal{L}_{\mathrm{CL}}= \lambda \hat{\mathcal{L}}_{\mathrm{CL}}^{\mathcal{P}} + (1-\lambda) \hat{\mathcal{L}}_{\mathrm{CL}}^{r}
\end{equation}
where $\lambda \in [0, 1]$ balances the effect of two aggregation functions.

\subsection{Edge Reconstruction Loss} 
While contrastive learning enforces consistency between different aggregation views of the same node, it does not explicitly model relational structures between nodes.
To complement this, we introduce an edge reconstruction loss that regularizes the learned representations to preserve observed cross-modal and intra-modal relationships in the graph. 

Specifically, we consider NL statement–Table edges and Table–Table edges in Sec.~\ref{sec:nodecontent}.
A subset of these edges is randomly masked using a Bernoulli distribution, $E_m \sim \mathrm{Bernoulli}(p)$, where $p < 1$ is the masking ratio.
The masked edges are treated as positive samples, while an equal number of negative samples $E^-$ are drawn from unconnected node pairs of the same types to match the corresponding edge types. For example, in Fig.~\ref{fig:framework}, two edges, $s_1$-$t_2$ and $s_1$-$t_5$, are sampled as masked positive edges, while two unconnected pairs, ($s_1$, $t_1$) and ($s_1$, $t_3$), are sampled as negative edges.

We define a decoder $h_\omega$ parameterized by $\omega$ to predict the probability that an edge exists between two nodes:
$$
\mathcal{L}^+ = {1 \over {|E_{m}|}} \sum_{{v, v'} \in E_{m}} \log h_{\omega}(\mathbf{h}^{\mathcal{P}}_{v}, \mathbf{h}^{\mathcal{P}}_{v'}),
$$
$$
\mathcal{L}^- = {1 \over {|E^-|}} \sum_{{v, v'} \in E^-} \log (1-h_{\omega}(\mathbf{h}^{\mathcal{P}}_{v}, \mathbf{h}^{\mathcal{P}}_{v'})),
$$
$$
\tilde{\mathcal{L}}_{\mathrm{ER}} = -(\mathcal{L}^+ + \mathcal{L}^-).
$$
Here, $
h_{\omega}(\mathbf{h}^{\mathcal{P}}_{v}, \mathbf{h}^{\mathcal{P}}_{v'}) = \mathrm{sigmoid}\left(\mathrm{MLP}(\mathbf{h}^{\mathcal{P}}_{v} \circ \mathbf{h}^{\mathcal{P}}_{v'})\right)
$
where MLP denotes a multi-layer perceptron and $\circ$ indicates element-wise multiplication.
We use meta-path–aware embeddings $\mathbf{h}^{\mathcal{P}}$ for edge reconstruction, as they explicitly encode semantic relational patterns aligned with the reconstructed edge types. 

In practice, the binary cross-entropy loss can exhibit large gradients when positive edges are assigned low predicted probabilities, which may destabilize training.
To improve robustness and smooth optimization, we apply a logarithmic transformation to the reconstruction loss:
\begin{equation}
\label{eq:erloss}
\mathcal{L}_{\mathrm{ER}} = \log(\tilde{\mathcal{L}}_{\mathrm{ER}} + 1).
\end{equation}

Finally, we jointly optimize the contrastive learning loss and the edge reconstruction loss:
\begin{equation}
\label{eq:loss}
\mathcal{L} = \mathcal{L}_{\mathrm{ER}} + \beta * \mathcal{L}_{\mathrm{CL}}
\end{equation}
where $\beta$ is a hyperparameter that controls the trade-off between the two losses. 
Appendix~\ref{ap:exam_loss} provides an example of the joint objective.

\setlength{\textfloatsep}{0pt}
\begin{algorithm}[t]
\renewcommand{\algorithmicrequire}{\textbf{Input:}}
\renewcommand\algorithmicensure{\textbf{Output:}}
\caption{Training Procedure}
\label{alg:training}
\begin{algorithmic}[1]
\Require Graph $G=(V,E)$ with $V=S\cup T$, training nodes $S^{\mathrm{train}}$ and $T^{\mathrm{train}}$, validation nodes $S^{\mathrm{val}}$ and $T^{\mathrm{val}}$;
\Ensure Optimized parameters $\Theta$;
\State Construct contrastive samples by sampling same-type negatives for each 
$v \in S^{\mathrm{train}} \cup T^{\mathrm{train}}$ for contrastive learning;
\State Obtain $E_{ST}$ including NL statement–Table edges and Table–Table edges incident to training nodes;
\State Sample masked edges $E_m \subset E_{ST}$ and negative edges $E^- \notin E_{ST}$;
\For{each epoch}
\State Compute
$\mathbf{h}_v^{\mathcal{P}}$ and $\mathbf{h}_v^{r}$ for anchors and their negatives;
\State Compute $\mathcal{L}_{\mathrm{CL}}$ by Eq.~\ref{eq:clloss};
\State Compute $\mathcal{L}_{\mathrm{ER}}$ on $E_m$ and $E^-$ by Eq.~\ref{eq:erloss};
\State $\mathcal{L} \gets \mathcal{L}_{\mathrm{ER}} + \beta * \mathcal{L}_{\mathrm{CL}}$ using Eq.~\ref{eq:loss};
\State Update the parameters $\Theta$ by Adam optimizer based on $\mathcal{L}$;
\State Keep the checkpoint with the best validation performance;
\EndFor
\Return $\Theta$.
\end{algorithmic}
\end{algorithm}

\smallskip
\noindent\textbf{Representation Materialization and Indexing.}
The overall training procedure is summarized in Alg.~\ref{alg:training}.
The model is trained by jointly optimizing the contrastive learning loss and the edge reconstruction loss.
After training, the learned node representations are fixed.
In particular, we use the embeddings derived from meta-path-based neighbor aggregation as semantic representations that encode relational context of target-type nodes. 

To enable efficient retrieval over large-scale data lakes, we materialize the learned table embeddings into an approximate nearest neighbor (ANN) index.
Specifically, we employ \emph{Annoy}~\cite{LiZSWLZL20}, a space-partitioning-based ANN method that provides a favorable trade-off between retrieval efficiency and accuracy, as demonstrated in prior work~\cite{EltabakhKEA23}.
The index is constructed and serves as the foundation for the adaptive retrieval module described in Sec.~\ref{sec:retrieval}.

\section{Experiments}
\label{sec:experiment}

In this section, we evaluate the effectiveness and efficiency of \ours{}.
We aim to answer the following questions.

\noindent$\bullet$ \textbf{Q1}. How does \ours{} perform on NL-based data discovery?
We evaluate \ours{} using NL queries with binary retrieval and ranked retrieval in terms of effectiveness and efficiency. (Sec.~\ref{sec:NL2table})

\noindent$\bullet$ \textbf{Q2}. How does \ours{} perform on table-based data discovery?
We assess the effectiveness and efficiency of \ours{} using tables as queries under binary retrieval and ranked retrieval. (Sec.~\ref{sec:table2table})

\noindent$\bullet$ \textbf{Q3}. How does \ours{} benefit downstream applications?
We evaluate the impact of \ours{} on an end-to-end table-based fact verification task, analyzing accuracy and runtime. (Sec.~\ref{sec:app})

\noindent$\bullet$ \textbf{Q4}. How do key factors affect the performance of \ours{}?
We study the sensitivity of \ours{} to training supervision and graph construction choices. (Sec.~\ref{sec:analysis})

\vspace{-0.6em}
\subsection{Experimental Setup}
\label{sec:setup}

\noindent\textbf{Datasets.}
We evaluate \ours{} on seven datasets supporting three research scenarios.
Their statistics are summarized in Table~\ref{tab:dataset}.

\noindent$\bullet$ \emph{Pharma}, \emph{ML-Open}, and \emph{UK-Open} are sourced from CMDL~\cite{EltabakhKEA23}, a state-of-the-art cross-modal data discovery system that supports both NL-based and table-based discovery.
These datasets were curated by the authors of CMDL and are widely used as benchmarks for evaluating cross-modal data discovery across diverse domains.
\emph{Pharma} contains tables from the DrugBank database~\cite{drugbank} together with a corpus of PubMed abstracts referenced by DrugBank entries.
For example, rows in tables such as ``Drug'' and ``Enzyme'' include citations to related abstracts.
\emph{ML-Open} consists of machine learning datasets, including both tabular and textual data collected from open data portals such as Kaggle and OpenML.
\emph{UK-Open} contains UK government open data in CSV format, paired with synthetic textual descriptions generated by the authors of CMDL.

\noindent$\bullet$ \emph{GeoQuery}~\cite{PalYKR23}, \emph{Atis}~\cite{PalYKR23}, and \emph{MMQA}~\cite{wummqa} are standard benchmarks for table question answering (QA).
\emph{GeoQuery} and \emph{Atis} focus on domain-specific settings -- U.S. geography and flight booking, respectively -- and have been widely adopted in prior work to evaluate table QA methods~\cite{PalYKR23}.
In contrast, \emph{MMQA} is a large-scale benchmark designed for comprehensive evaluation of table QA across domains.
It spans 200 databases covering 138 domains, each containing multiple interrelated tables, and poses significant challenges for cross-domain reasoning and complex question understanding.

\noindent$\bullet$ \emph{TabFact}~\cite{ChenWCZWLZW20} is a widely used benchmark for table-based fact verification.
It consists of Wikipedia tables paired with crowd-sourced factual claims labeled as either supported or refuted.
Our experiments focus on simple claims that describe single-table rows, as complex claims are often artificially constructed and tend to reflect factual expressions less faithfully than straightforward claims~\cite{ChaiGZF021}.

\setlength{\textfloatsep}{0pt}
\begin{table}[t]
\setlength{\abovecaptionskip}{0cm}
\setlength{\belowcaptionskip}{0cm}
\footnotesize
\centering
\caption{Statistical properties for the datasets}
\label{tab:dataset}
\begin{tabular}{cccc}
\toprule
Datasets&\# NL statements& \# tables& Query modality\\
\midrule
Pharma&1380&82&NL statement, table\\
ML-Open&1501&158&NL statement, table\\
UK-Open&2360&654&NL statement\\
GeoQuery&832&7&NL statement\\
Atis&515&11&NL statement\\
MMQA&3313&553&NL statement\\
TabFact&49850&9181&NL statement\\
\bottomrule
\end{tabular}
\end{table}


\smallskip
\noindent\textbf{Evaluation Metrics.}
For binary retrieval, we evaluate effectiveness using \emph{precision}, \emph{recall}, and \emph{F1 score}.
For ranked retrieval, we report \emph{precision@$k$} and \emph{recall@$k$} in terms of effectiveness.
We evaluate efficiency in terms of \emph{training time}, \emph{training memory usage}, and \emph{average inference time per query}. 

\setlength{\textfloatsep}{0pt}
\begin{table*}[t]
\setlength{\tabcolsep}{2.5pt}
\setlength{\abovecaptionskip}{0cm}
\setlength{\belowcaptionskip}{0cm}
\footnotesize
\centering
\caption{Effectiveness of binary retrieval for NL-based data discovery (P represents precision, R represents recall)
}
\label{tab:NL2table}
\begin{tabular}{cccccccccccccccccccccc}
\toprule
\multirow{2}{*}{Method}&\multicolumn{3}{c}{Pharma}&\multicolumn{3}{c}{ML-Open}&\multicolumn{3}{c}{UK-Open}&\multicolumn{3}{c}{GeoQuery}&\multicolumn{3}{c}{Atis}&\multicolumn{3}{c}{MMQA}&\multicolumn{3}{c}{TabFact}\\
\cmidrule(lr){2-4} \cmidrule(lr){5-7} \cmidrule(lr){8-10} \cmidrule(lr){11-13} \cmidrule(lr){14-16} \cmidrule(lr){17-19} \cmidrule(lr){20-22}
&P&R&F1&P&R&F1&P&R&F1&P&R&F1&P&R&F1&P&R&F1&P&R&F1\\
\midrule
CMDL&0.82&0.996&0.683&0.905&0.537&0.674&0.507&0.852&0.636&0.501&0.802&0.616&0.499&0.898&0.642&0.463&0.255&0.329&0.428&0.163&0.236\\
\ours&\textbf{0.974}&\textbf{1}&\textbf{0.987}&\textbf{0.967}&\textbf{1}&\textbf{0.983}&\textbf{0.987}&\textbf{0.998}&\textbf{0.992}&\textbf{0.625}&\underline{0.878}&\textbf{0.73}&\textbf{0.772}&\underline{0.934}&\textbf{0.845}&\textbf{0.663}&\underline{0.772}&\textbf{0.713}&\underline{0.665}&\textbf{0.994}&\underline{0.797}\\
\ours$_\text{linear}$&0.958&1&0.978&0.94&1&0.969&0.978&0.989&0.983&0.615&0.862&0.718&0.705&0.802&0.751&0.623&0.708
&0.663&0.638&0.968
&0.769\\
\ours$_\text{noCL}$&\underline{0.968}&1&\underline{0.984}&\underline{0.956}&1&\underline{0.978}&\underline{0.984}&\underline{0.996}&\underline{0.99}&\underline{0.62}&0.87&\underline{0.724}&\underline{0.762}&0.897&\underline{0.824}&\underline{0.645}&0.756&\underline{0.696}&0.649&0.97
&0.778\\
\ours$_\text{noER}$&0.96&1&0.979&0.951&1&0.975&0.98&0.99&0.985&0.619&0.866
&0.722&0.756&0.891&0.818&0.637&0.757&0.692&0.657&0.962&0.781\\
Llama4-Scout&0.822&1&0.902&0.815&1&0.898&0.754&0.967&0.847&0.502
&\textbf{0.975}
&0.663&0.591
&\textbf{0.941}&0.726&0.439&\textbf{0.863}&0.582&0.518&\underline{0.985}&0.679\\
MTR-TableLlama&N/A&N/A&N/A&N/A&N/A&N/A&N/A&N/A&N/A&0.478
&0.827&0.606&0.528
&0.834&0.647&0.512&0.355&0.419&N/A&N/A&N/A\\
PASTA&N/A&N/A&N/A&N/A&N/A&N/A&N/A&N/A&N/A&N/A&N/A&N/A&N/A&N/A&N/A&N/A&N/A&N/A&\textbf{0.731}&0.899&\textbf{0.806}\\
\bottomrule
\end{tabular}\\
\vspace{-1em}
\end{table*}

\smallskip
\noindent\textbf{Implementation Details.}
By default, each dataset is split into training, validation, and test sets with a ratio of 20\% / 20\% / 60\%.
We set $K=10$ for top-$K$ Table-Table and NL statement-NL statement neighbor construction.
We initialize content features for both NL statements and tables using a pretrained BERT model~\cite{DevlinCLT19}.
The node embedding dimension is set to 128.
We train the model for 50 epochs using Adam with a learning rate of 0.0008.
Unless otherwise specified, we set $r=3$, $\tau=0.07$, $\lambda=0.5$, and $\beta=0.5$ for the cross-modal representation module.
For binary retrieval, the candidate size is set to 50\% of the total number of tables in each dataset.
All experiments are conducted on a server running Red Hat Enterprise Linux with a 2.60GHz Intel\textregistered\ Xeon\textregistered\ CPU, 512GB memory, and two NVIDIA Tesla P100 GPUs (16GB each).

\begin{figure*}[t]
\setlength{\abovecaptionskip}{0cm}
\setlength{\belowcaptionskip}{-0.5cm}
\centering
\subfigtopskip=0pt
\subfigbottomskip=2pt
\subfigcapskip=-6pt
\includegraphics[height=0.47cm]{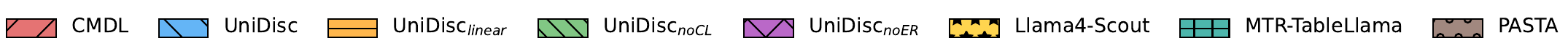}\\
\subfigure[Training time]{\label{subfig:training} \includegraphics[width=56mm]{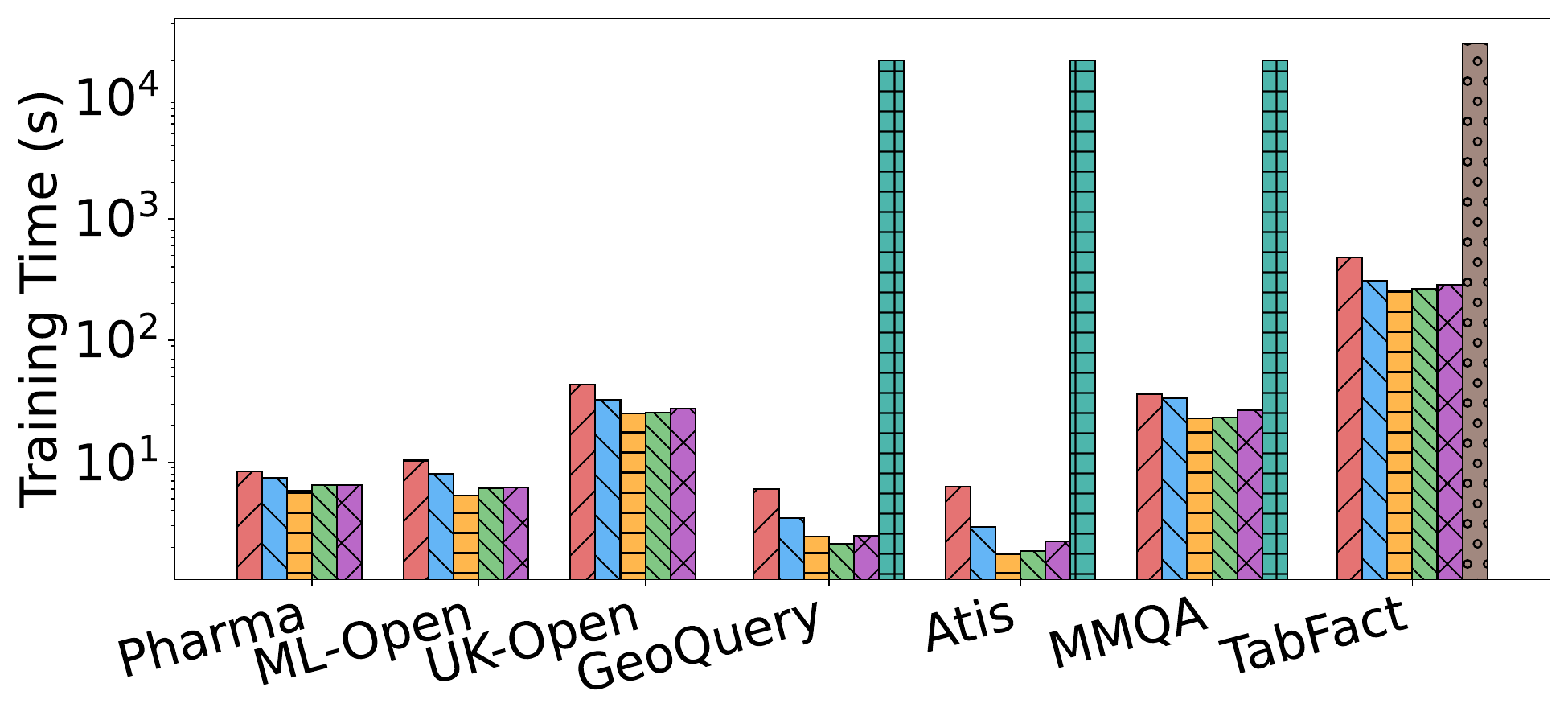}}
\subfigure[Training memory usage]{\label{subfig:memory} \includegraphics[width=56mm]{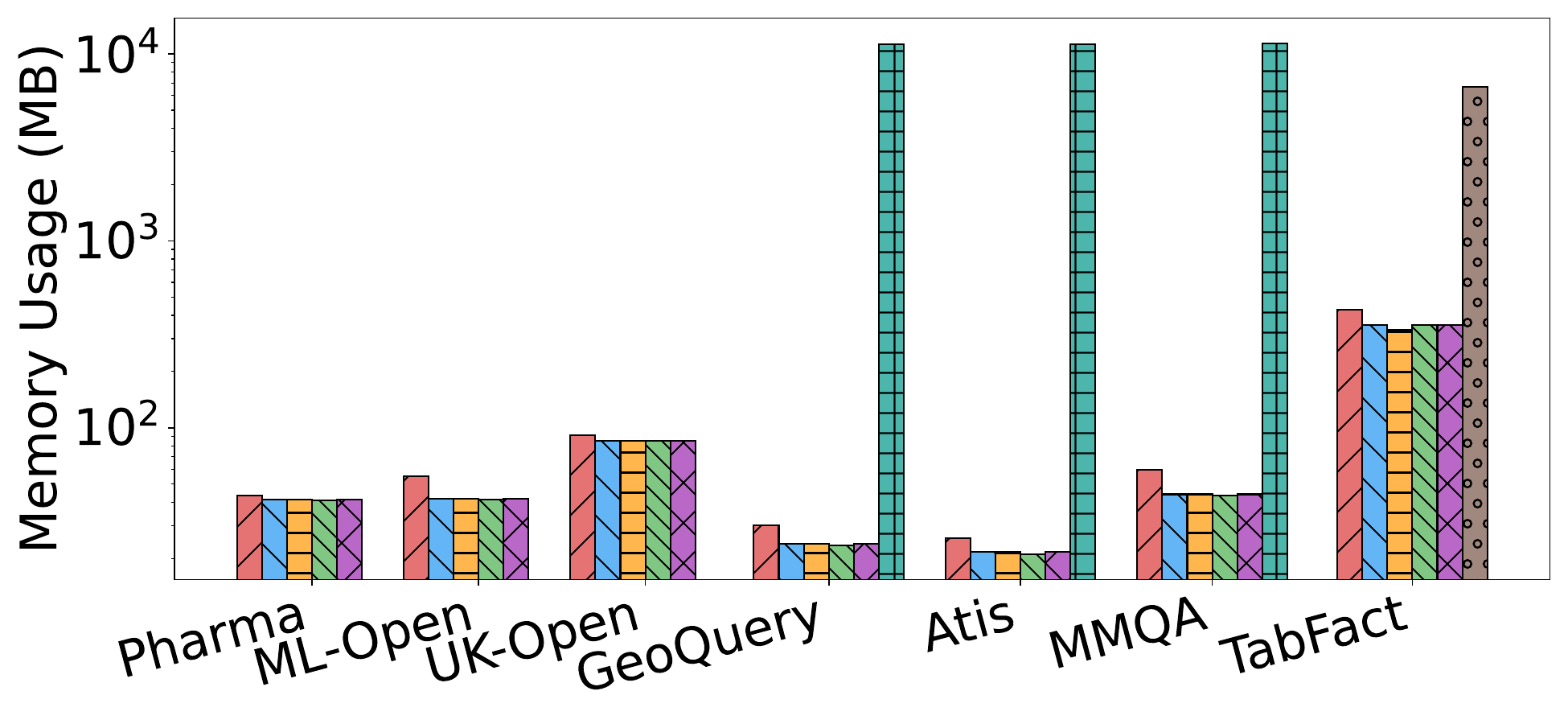}}
\subfigure[Average inference time per query]{\label{subfig:inference} \includegraphics[width=56mm]{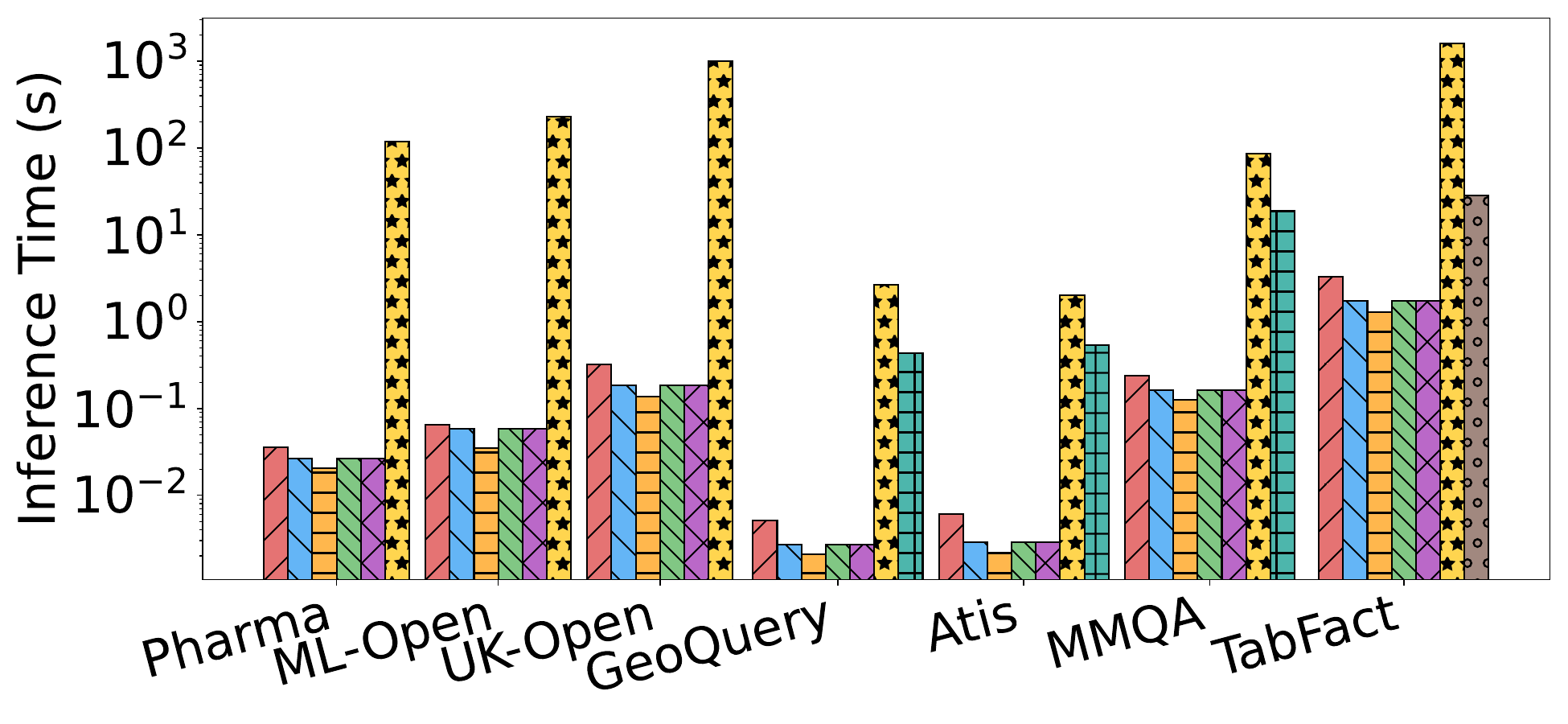}}
\caption{\label{fig:efficiency}{Efficiency of binary retrieval for NL-based data discovery.}
}
\end{figure*}

\vspace{-0.5em}
\subsection{NL-based Data Discovery (Q1)}
\label{sec:NL2table}

We compare \ours{} with the following baselines:

\noindent$\bullet$ \emph{Variants of \ours}. 
We consider three variants:
\ours$_\text{linear}$ replaces the Bi-LSTM node content encoder with a linear transformation;
\ours$_\text{noCL}$ removes the contrastive learning loss; and
\ours$_\text{noER}$ removes the edge reconstruction loss.

\noindent$\bullet$ \emph{CMDL}~\cite{EltabakhKEA23}. 
A state-of-the-art cross-modal data discovery system supporting both NL-based and table-based discovery.

\noindent$\bullet$ \emph{Llama4-Scout}~\cite{llama4}. 
A large language model with an industry-leading 10M-token context window.

\noindent$\bullet$ \emph{MTR-TableLlama}~\cite{wummqa}. 
A state-of-the-art method identifies relevant tables for multi-table QA.
It is the first to explicitly perform multi-table retrieval as part of an end-to-end QA pipeline.

\noindent$\bullet$ \emph{Pneuma}~\cite{balaka2025pneuma}. 
A state-of-the-art method for retrieving the most relevant table in single-table QA.

\noindent$\bullet$ \emph{PASTA}~\cite{GuF0NZ022}. 
A state-of-the-art method for table-based fact verification that predicts whether a table supports or refutes a claim.
We adapt it for relevance prediction by fine-tuning on TabFact with relabeled (claim, table) pairs.

\smallskip
\noindent
We compare \ours{} with CMDL and Llama4-Scout on all seven datasets, as they support general-purpose NL-based data discovery.
MTR-TableLlama and Pneuma are evaluated on the three table QA datasets (Atis, GeoQuery, and MMQA), as they are specifically designed for table QA.
PASTA is evaluated on the TabFact dataset, which focuses on table-based fact verification.

\subsubsection{Binary Retrieval}
\label{sec:set_retrieval}

We first evaluate the effectiveness and efficiency of \ours{} under binary retrieval.

\smallskip
\noindent\textbf{Task Implementation.}
Binary retrieval requires predicting binary relevance for each (query, table) pair.
Accordingly, we adapt applicable baselines to this setting.
The adapted \emph{PASTA} directly predicts 
binary relevance for (query, table) pairs.
For \emph{Llama4-Scout}, we adopt a prompting-based relevance prediction strategy following prior work on LLM-based retrieval~\cite{ArabzadehC25a}, where the model is prompted to determine whether a table is relevant to a query and outputs a binary label.
\emph{CMDL} is originally designed to produce column-level embeddings.
To enable table-level relevance prediction, we apply mean pooling over column embeddings to obtain table representations and use our binary retrieval method.
Since the code of \emph{MTR-TableLlama} is not publicly available, we re-implement it by following the pipeline suggested by the original authors~\cite{wummqa}.
Specifically, each query is decomposed into sub-questions using GPT-4-turbo~\cite{gpt4} with the prompting strategy in~\cite{wummqa}, augmented with examples from~\cite{abs-2402-11166}, and both tables and sub-questions are encoded by fine-tuning TableLlama-7B~\cite{ZhangYL024} with LoRA~\cite{HuSWALWWC22} on WikiTableQuestions~\cite{PasupatL15}. To support binary retrieval, we represent each query by averaging the embeddings of its sub-questions and apply our binary retrieval method.
We do not include \emph{Pneuma}, as it is specifically designed for top-$k$ retrieval by combining multiple relevance signals and does not support binary relevance prediction.

\smallskip
\noindent\textbf{Effectiveness}.
Table~\ref{tab:NL2table} summarizes effectiveness results. Observe that: 
(1) \ours{} and its variants consistently outperform CMDL across all datasets in precision, recall, and F1, with improvements up to 200\%.
Although Llama4-Scout achieves higher recall on some datasets, \ours{} yields better precision and F1, suggesting that they can generalize across domains and effectively model the relationship between NL statements and tables.
(2) On table QA datasets (Atis, GeoQuery, MMQA), \ours{} surpasses MTR-TableLlama by up to 70\%, demonstrating stronger retrieval of question-relevant tables.
(3) \ours{} is competitive with PASTA on TabFact, showing its ability to support fact verification.
(4) \ours{} outperforms \ours$_\text{linear}$, confirming the benefit of Bi-LSTM encoders for capturing contextual signals across columns.
(5) \ours{} outperforms \ours$_\text{noCL}$ and \ours$_\text{noER}$, validating that the joint optimization objective leads to more robust and discriminative node embeddings. 

\smallskip
\noindent\textbf{Efficiency}. We evaluate the efficiency of the compared methods across three key aspects of training time, training memory usage, and average inference time per query. Note that Llama4-Scout is used in a zero-shot setting via prompting and is not fine-tuned, and thus we report only its inference time.

\noindent$\bullet$ \emph{Training time}.
For PASTA, training time refers to the average per-epoch fine-tuning time. For MTR-TableLlama, it includes both the average per-epoch fine-tuning time and the additional MLP training time. For \ours{}, its variants, and CMDL, it includes the average per-epoch model training time plus MLP overhead. As shown in Fig.~\ref{subfig:training}:
(1) \ours{} and its variants consistently require less training time compared to other methods across all datasets.
Compared to the strongest baseline, CMDL, \ours{} achieves at least a 10\% reduction in training time and up to a 2$\times$ speedup in some cases, demonstrating high computational efficiency.
(2) \ours{} is three orders of magnitude faster than MTR-TableLlama and an order of magnitude faster than PASTA, reflecting the substantial overhead of large pretrained language models used in MTR-TableLlama and PASTA. 
(3) \ours$_{\text{linear}}$ is the fastest variant, as it replaces the Bi-LSTM with a linear encoder and thus avoids modeling inter-content dependencies and contextual information. Similarly, \ours$_{\text{noCL}}$ and \ours$_{\text{noER}}$ train faster than \ours{}, since optimizing a single loss is less computationally intensive than joint optimization.

\noindent$\bullet$ \emph{Training memory usage}.
As shown in Fig.~\ref{subfig:memory}, \ours{}, its variants, and CMDL consume substantially less memory than other methods, with CMDL requiring up to 30\% more memory than \ours{}. 
In contrast, MTR-TableLlaMa and PASTA incur significantly higher memory usage -- by three and one orders of magnitude, respectively -- due to the storage demands of their large language models.

\noindent$\bullet$ \emph{Inference time}. PASTA and Llama4-Scout evaluate each query against all tables to predict binary relevance, whereas other methods adopt our binary retrieval method.
For fairness, we report inference time as the amortized per-query cost, computed as the total time of (i) one-time table embedding generation and (ii) executing table relevance evaluation for all queries, divided by the number of queries.
As shown in Fig.~\ref{subfig:inference}:
(1) \ours{} and its variants consistently achieve the lowest inference time, with at least a 10\% reduction and up to a 2$\times$ speedup over the strongest baseline, CMDL. 
(2) Llama4-Scout exhibits the highest inference time, being at least three orders of magnitude slower due to its 17B-parameter model and long input prompts.
(3) MTR-TableLlaMa and PASTA also incur substantially higher inference costs -- approximately two and one orders of magnitude more, respectively -- owing to their more complex model architectures for embedding generation.
(4) \ours{}, \ours$_\text{noCL}$, and \ours$_\text{noER}$ have identical inference costs, as they share the same model architecture.

\smallskip
\noindent\textbf{Discussion}. In summary, \ours{} and its variants consistently outperform competing methods in both effectiveness and efficiency.
\ours{} achieves the best overall effectiveness, while \ours$_\text{linear}$ is the most efficient, enabling flexible trade-offs between performance and computational cost.
Although Llama4-Scout and PASTA attain slightly higher effectiveness in some cases, \ours{} offers a more favorable balance between effectiveness and efficiency.

Across datasets, Pharma is relatively easier due to stronger lexical overlap and explicit links between NL statements and tables (e.g., DrugBank citations), resulting in smaller performance gaps.
In contrast, MMQA and TabFact are more challenging due to weaker lexical alignment, requiring models to leverage contextual and structural signals.
Despite this increased difficulty,  \ours{} remains the top-performing model by effectively exploiting graph structure and multi-type neighborhood information.


\begin{figure}[t]
\setlength{\abovecaptionskip}{0cm}
\setlength{\belowcaptionskip}{0cm}
\centering
\subfigtopskip=0pt
\subfigbottomskip=2pt
\subfigcapskip=-4pt
\includegraphics[height=0.5cm]{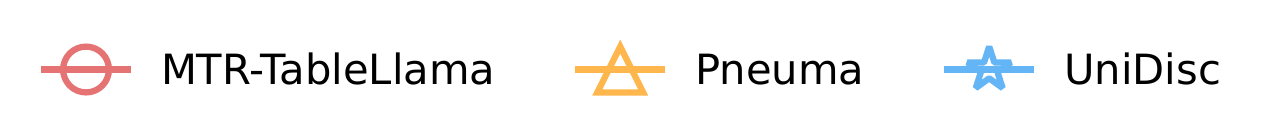}\\
\subfigure[GeoQuery]{\label{subfig:geoquery} \includegraphics[width=26mm, height=20mm]{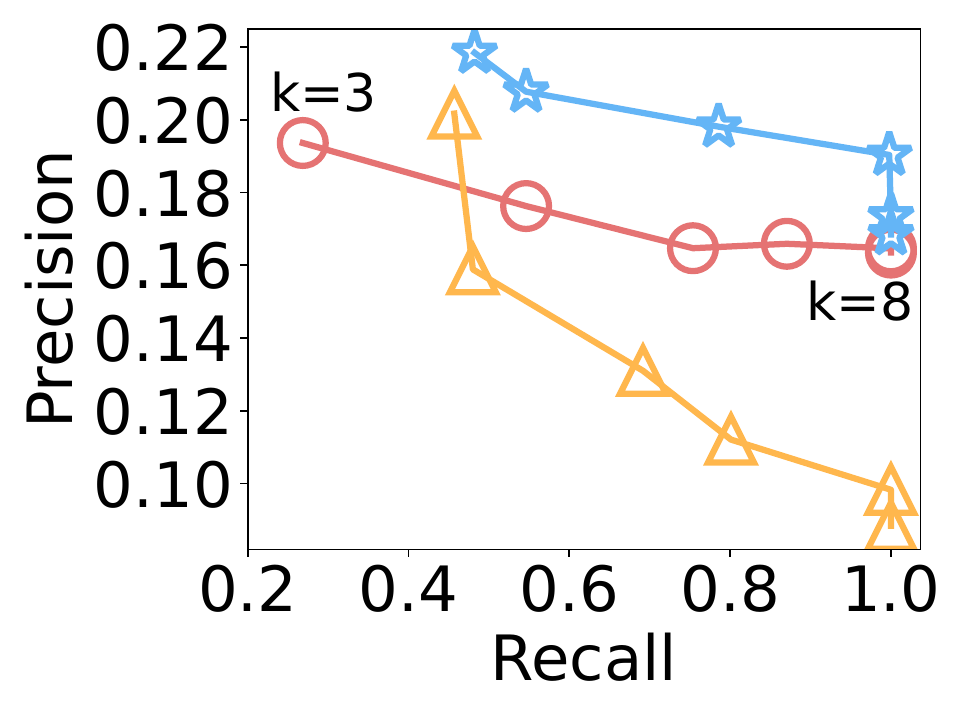}}
\subfigure[Atis]{\label{subfig:atis} \includegraphics[width=26mm, height=20mm]{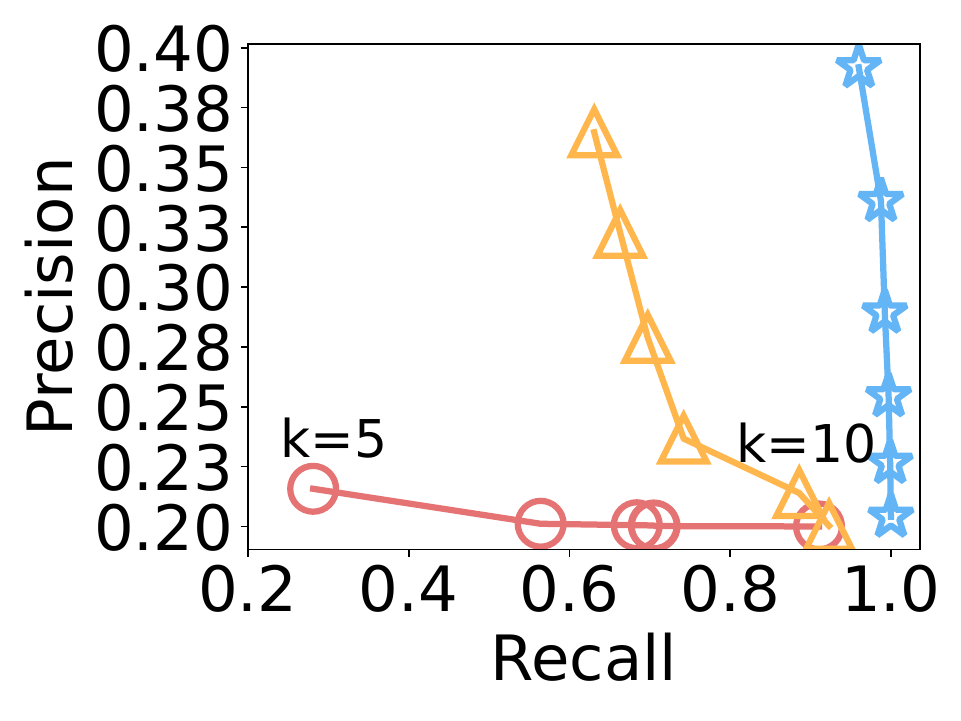}}
\subfigure[MMQA]{\label{subfig:MMQA} \includegraphics[width=27mm, height=20mm]{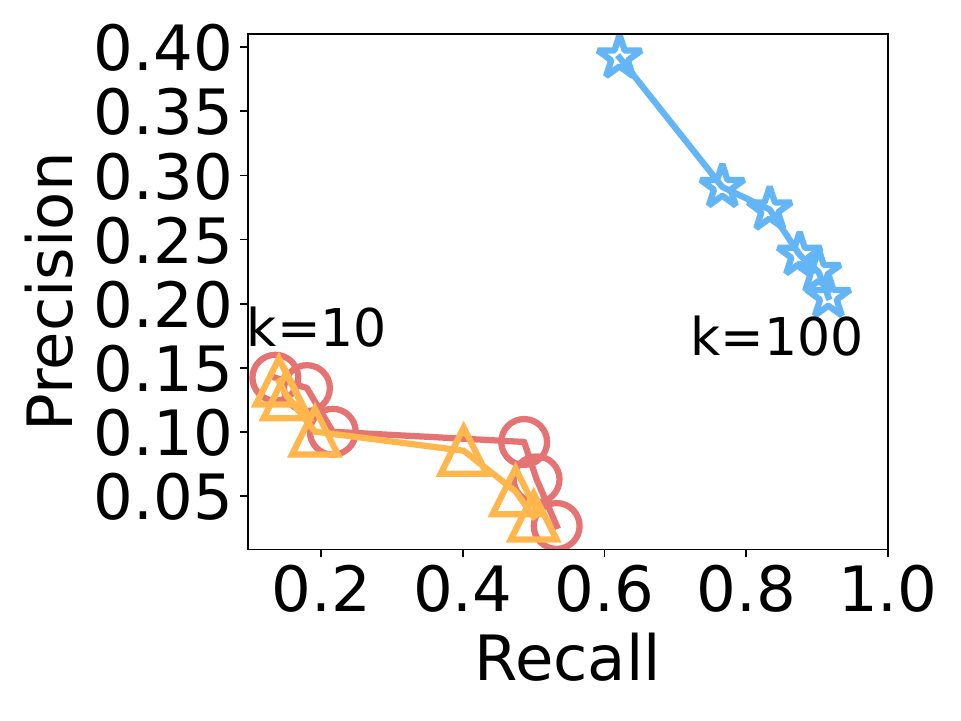}}\\
\subfigure[GeoQuery]{\label{subfig:tgeoquery} \includegraphics[width=26mm, height=20mm]{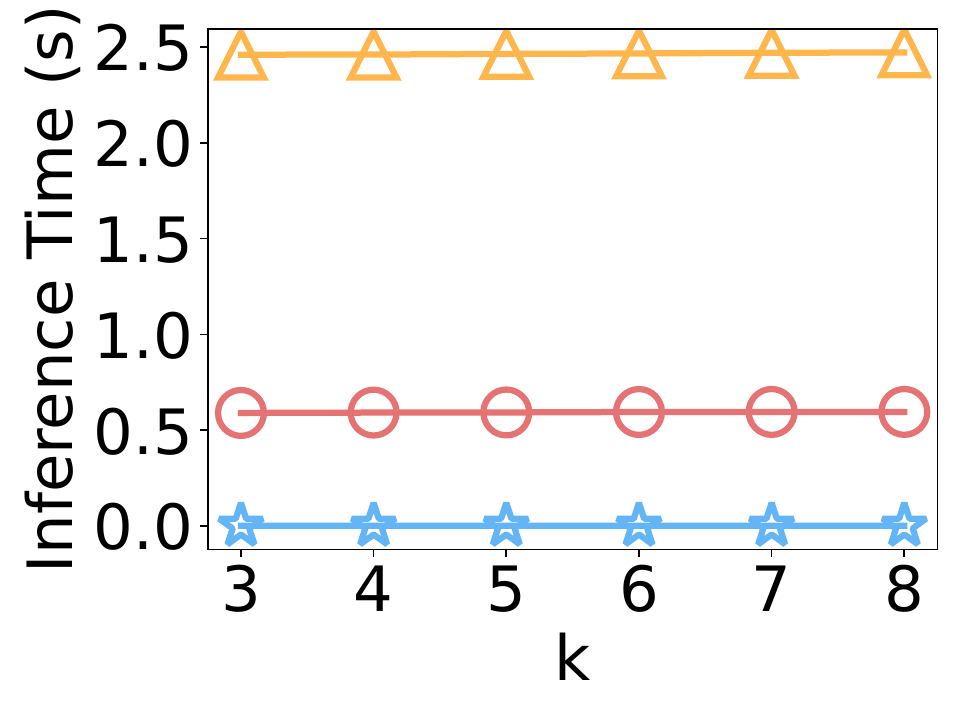}}
\subfigure[Atis]{\label{subfig:tatis} \includegraphics[width=26mm, height=20mm]{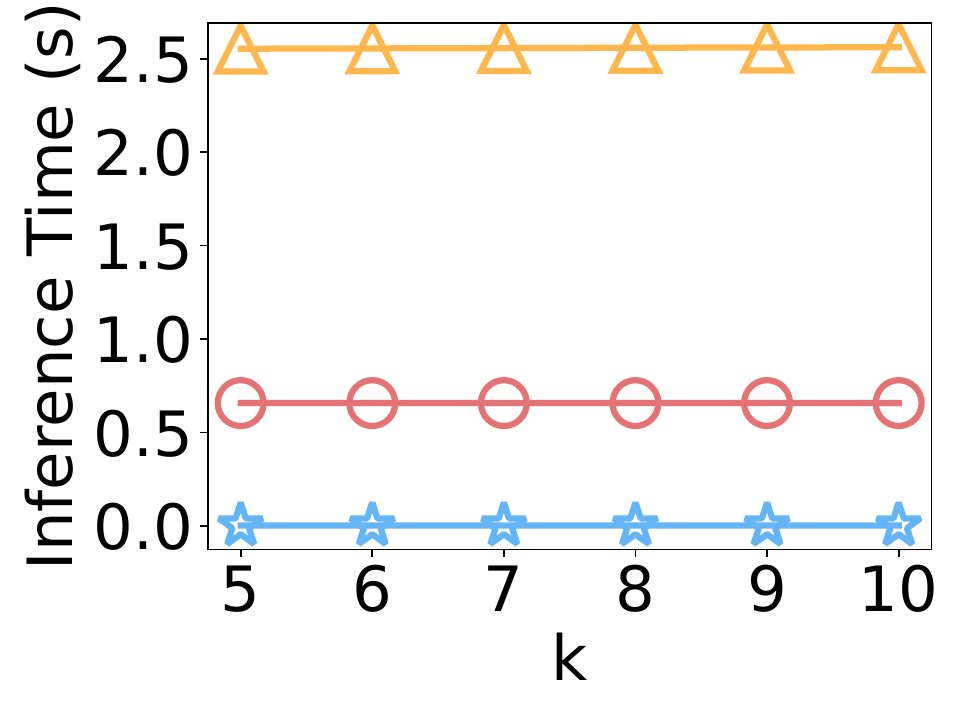}}
\subfigure[MMQA]{\label{subfig:tMMQA} \includegraphics[width=27mm, height=20mm]{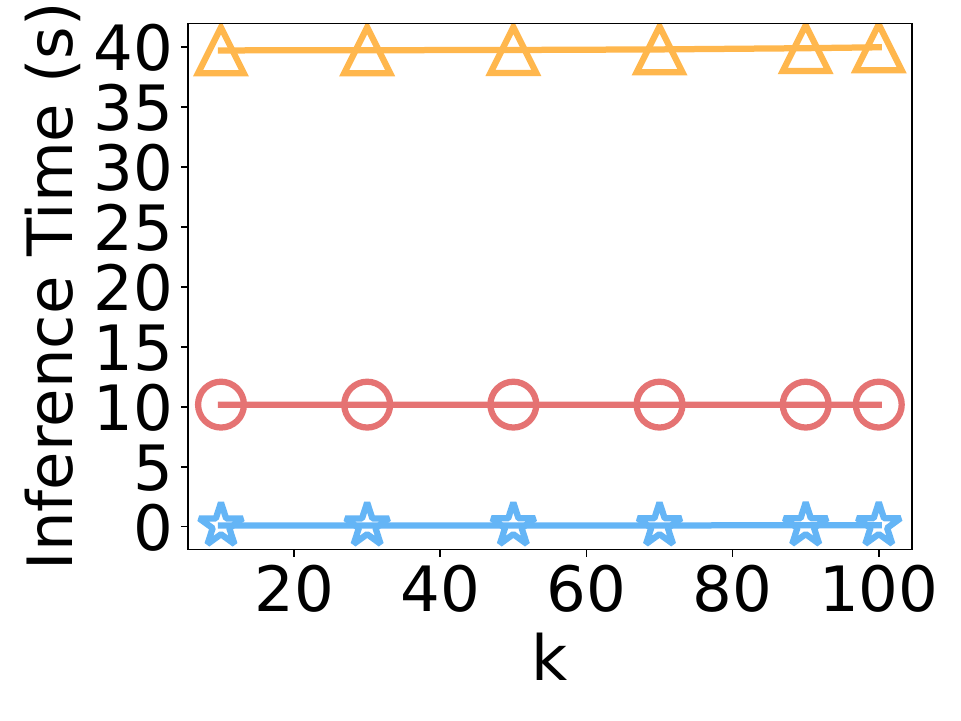}}
\caption{\label{fig:QA}
{Evaluation of ranked retrieval for NL-based data discovery. 
}
}
\end{figure}

\subsubsection{Ranked Retrieval}
\label{sec:topk_retrieval}

In this section, we evaluate the effectiveness and efficiency of \ours{} under ranked retrieval.

\smallskip
\noindent\textbf{Task Implementation.}
We consider three table QA datasets -- Atis, GeoQuery, and MMQA, as table question answering is a representative downstream application for ranked retrieval in prior work~\cite{WangF23,balaka2025pneuma,wummqa}. Accordingly, we compare \ours{} against two state-of-the-art table retrieval methods for table QA: MTR-TableLlama and Pneuma.
{Pneuma} follows its original pipeline and retrieves the top-$k$ most relevant tables for single-table QA, outperforming earlier methods such as Solo~\cite{WangF23}.
{MTR-TableLlama} performs multi-table retrieval by decomposing each input question into multiple sub-questions, encoding them using a fine-tuned TableLlama model (as described in Sec.~\ref{sec:set_retrieval}), and retrieving the top-$k$ most relevant tables via $n$ independent retrieval rounds, one for each sub-question.

The range of $k$ in ranked retrieval varies across datasets due to differences in the number of available tables.
As recall approaches one when nearly all tables are retrieved, excessively large $k$ values become uninformative.
Accordingly, we set the maximum $k$ to 10 for Atis and 8 for GeoQuery, which are close to their respective table counts.
For MMQA, which contains 553 tables, \ours{} reaches full recall at $k=100$; we therefore set the maximum $k$ to 100.

\smallskip
\noindent\textbf{Effectiveness.}
As shown in Figs.~\ref{subfig:geoquery}–\ref{subfig:MMQA}, \ours{} consistently outperforms MTR-TableLlaMa and Pneuma in both precision and recall under ranked retrieval.
In particular, \ours{} reaches full recall at smaller values of $k$.
This indicates that \ours{} more effectively identifies question-relevant tables when only a limited number of tables are returned, thereby increasing the likelihood of accurate answer generation in table QA.

\smallskip
\noindent\textbf{Efficiency.}
Since Pneuma does not involve model training and directly invokes LLMs for table representation, we focus on inference time.
As these methods adopt different embedding generation strategies, inference time is reported as the amortized per-query cost, computed by dividing the sum of the one-time table embedding generation and the total retrieval time by the number of queries.
As illustrated in Figs.~\ref{subfig:tgeoquery}–\ref{subfig:tMMQA}, \ours{} achieves up to two orders of magnitude speedup over MTR-TableLlaMa and up to three orders of magnitude speedup over Pneuma, with increasing $k$ having a negligible impact on inference time.
This is because inference is dominated by embedding generation rather than retrieval.
\ours{} generates embeddings substantially more efficiently than MTR-TableLlaMa and Pneuma, both of which rely on large foundation models.
For example, on MMQA, \ours{} requires about 100ms for embedding generation and only a few milliseconds for retrieval, whereas MTR-TableLlaMa and Pneuma require approximately 10s and 39s for embedding generation, respectively; table retrieval for both methods takes only milliseconds to sub-second time.



\setlength{\textfloatsep}{0pt}
\begin{table}[t]
\setlength{\tabcolsep}{3pt}
\setlength{\abovecaptionskip}{0cm}
\setlength{\belowcaptionskip}{0cm}
\footnotesize
\centering
\caption{Effectiveness of binary retrieval for table-based data discovery 
}
\label{tab:table2table}
\begin{tabular}{ccccccc}
\toprule
\multirow{2}{*}{Method}&\multicolumn{3}{c}{Pharma}&\multicolumn{3}{c}{ML-Open}\\
\cmidrule(lr){2-4} \cmidrule(lr){5-7}
&P&R&F1&P&R&F1\\
\midrule
CMDL&0.815&0.856&0.835&0.865&0.832&0.848\\
\ours&\textbf{0.956}&\textbf{1}&\textbf{0.977}&\textbf{0.965}&\textbf{1}&\textbf{0.982}\\
\ours$_\text{linear}$&0.939&1&0.968&0.942&1&0.97
\\
\ours$_\text{noCL}$&0.948&1&0.973&0.952&1&0.975\\
\ours$_\text{noER}$&\underline{0.95}&1&\underline{0.974}&\underline{0.958}&1&\underline{0.979}\\
\bottomrule
\end{tabular}
\vspace{-0.4cm}
\end{table}

\vspace{-0.3cm}
\subsection{Table-based Data Discovery (Q2)}
\label{sec:table2table}

We compare \ours{} with CMDL~\cite{EltabakhKEA23} on Pharma and ML-Open, as CMDL supports table-based discovery and provides ground-truth (table, table) relevance labels for these two datasets, enabling a fair comparison.

\subsubsection{Binary Retrieval}

We adapt CMDL to this setting by following the binary retrieval setup used for NL-based discovery in Sec.~\ref{sec:set_retrieval}.

\begin{figure}[t]
\setlength{\abovecaptionskip}{0cm}
\setlength{\belowcaptionskip}{-0.5cm}
\centering
\includegraphics[height=0.43cm]{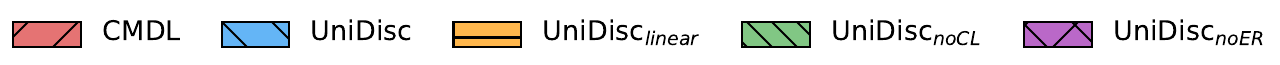}\\
\includegraphics[width=55mm]{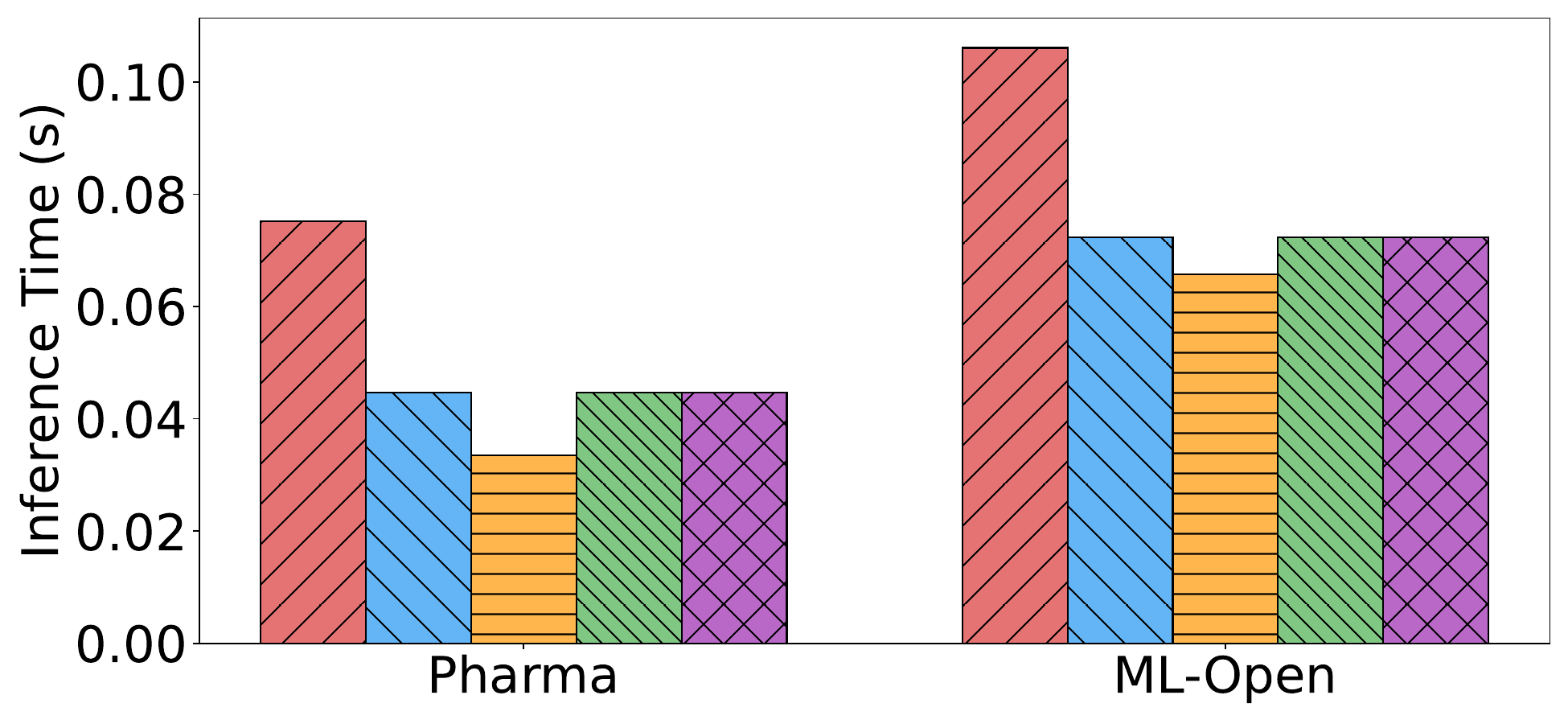}
\caption{\label{fig:t2t}{Average inference time per query of binary retrieval for table-based data discovery.}
}
\end{figure}

\begin{figure}[t]
\setlength{\abovecaptionskip}{0cm}
\setlength{\belowcaptionskip}{0cm}
\centering
\subfigtopskip=0pt
\subfigbottomskip=0pt
\subfigcapskip=-6pt
\includegraphics[height=0.5cm]{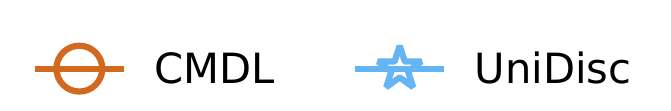}\\
\subfigure[Pharma]{\label{subfig:pharma} \includegraphics[width=32mm, height=22mm]{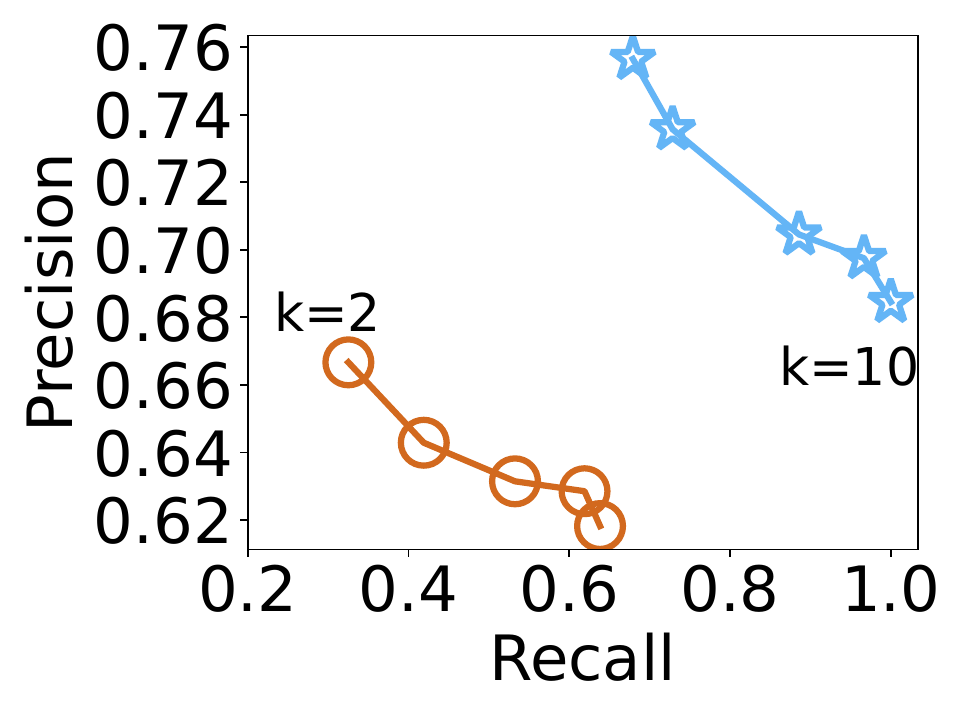}}
\subfigure[ML-Open]{\label{subfig:mlopen} \includegraphics[width=32mm, height=22mm]{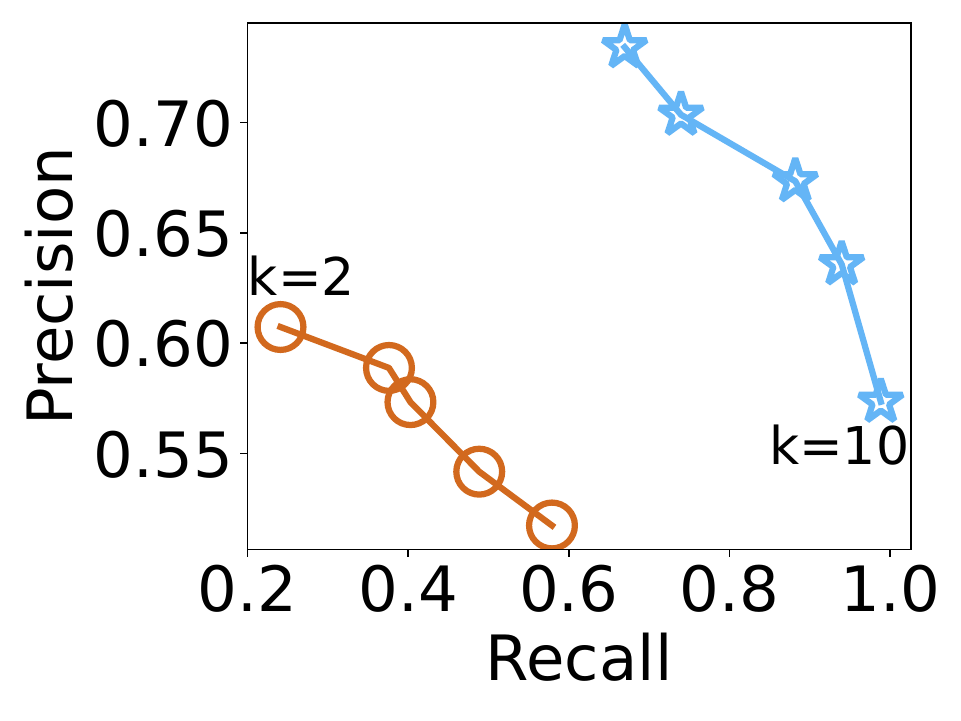}}\\
\subfigure[Pharma]{\label{subfig:tpharma} \includegraphics[width=32mm, height=22mm]{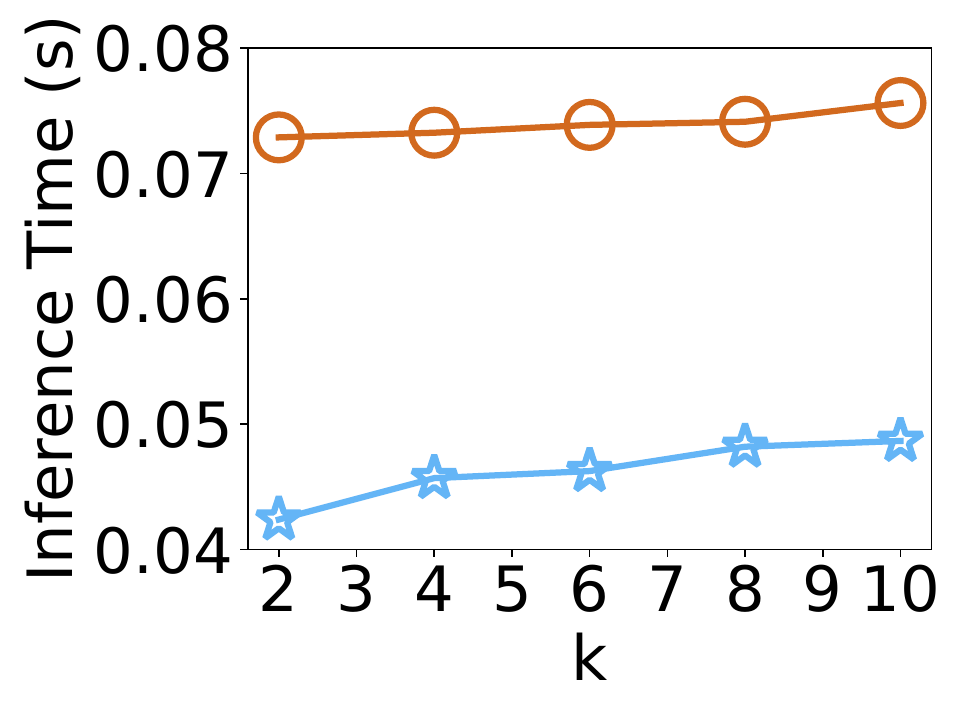}}
\subfigure[ML-Open]{\label{subfig:tmlopen} \includegraphics[width=32mm, height=22mm]{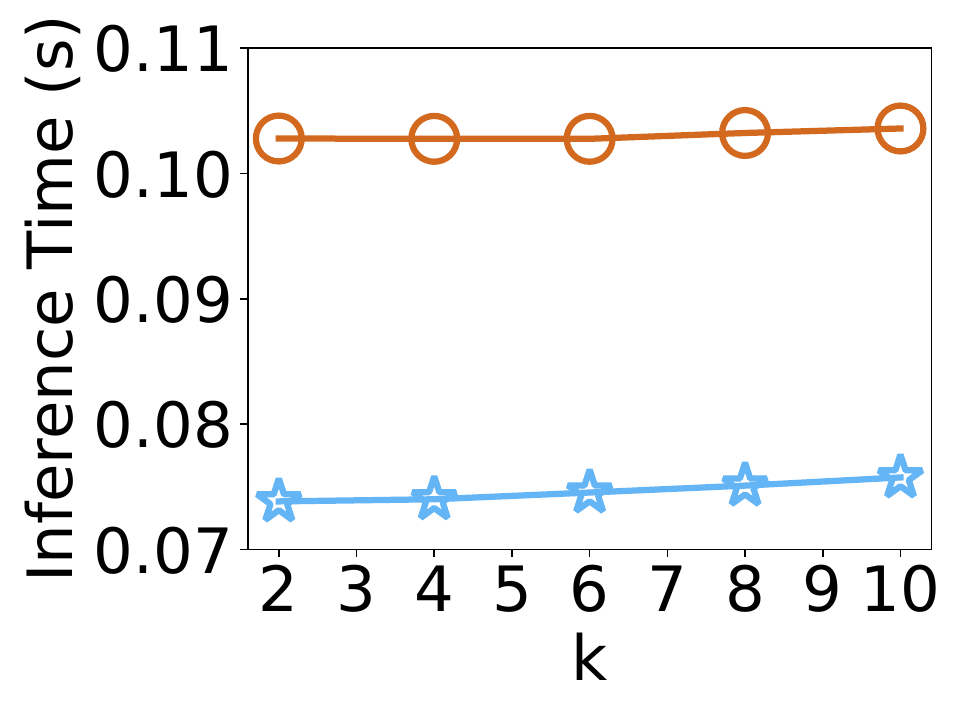}}
\caption{\label{fig:topk_table}
{Evaluation of ranked retrieval for table-based data discovery. 
}
}
\end{figure}

\smallskip
\noindent\textbf{Effectiveness.}
The effectiveness results are reported in Table~\ref{tab:table2table}.
\ours{} and its variants consistently outperform CMDL on both datasets, achieving up to a 17\% improvement.
This demonstrates \ours{}'s ability to learn discriminative table representations that bring semantically related tables closer while separating irrelevant ones in the embedding space.
\ours{} also achieves better performance than all variants -- \ours$_{\text{linear}}$, \ours$_{\text{noCL}}$, and \ours$_{\text{noER}}$ -- highlighting the benefits of using a Bi-LSTM encoder for table content and jointly optimizing the contrastive learning and edge reconstruction objectives.
Together, these design choices contribute to more informative and robust table representations.

\smallskip
\noindent\textbf{Efficiency.}
Since training time and memory usage have been reported in Figs.~\ref{subfig:training} and~\ref{subfig:memory}, we focus on average inference time per query, as shown in Fig.~\ref{fig:t2t}.
Consistent with the NL-based data discovery setting, \ours{} and its variants achieve at least 30\% lower inference latency than CMDL across both datasets.

\subsubsection{Ranked Retrieval}
To support ranked retrieval with CMDL, we follow its original design by computing relevance scores between the query embedding and individual column embeddings of each candidate table, and aggregating these scores to rank tables.

\smallskip
\noindent\textbf{Results.}
As shown in Figs.~\ref{subfig:pharma}–\ref{subfig:mlopen}, \ours{} consistently achieves higher precision and recall than CMDL across different values of $k$, and reaches full recall at $k=10$, demonstrating its effectiveness for table-based data discovery under ranked retrieval.
These results indicate that even when users request only a limited number of tables, \ours{} is more likely to retrieve the relevant ones.

In terms of efficiency, Figs.~\ref{subfig:tpharma}–\ref{subfig:tmlopen} further show that \ours{} incurs lower inference latency than CMDL, achieving up to a 75\% speedup, which further highlights the efficiency of our framework.

\setlength{\textfloatsep}{0pt}
\begin{table}[t]
\setlength{\tabcolsep}{3pt}
\setlength{\abovecaptionskip}{0cm}
\setlength{\belowcaptionskip}{0cm}
\footnotesize
\centering
\caption{End-to-end evalution for table-based fact verification 
}
\label{tab:TFV}
\begin{tabular}{ccccc}
\toprule
Methods&Precision&Recall&F1&Total inference time\\
\midrule
OpenTFV&0.618&0.687&0.651&2261s
\\
\ours{}&\textbf{0.623}&\textbf{0.688}&\textbf{0.654}&\textbf{13s}\\
\bottomrule
\end{tabular}
\end{table}

\subsection{Downstream Application (Q3)}
\label{sec:app}

\ours{} is designed as a fundamental building block for downstream tasks that rely on effective data discovery.
In this section, we use table-based fact verification as a representative application to illustrate how \ours{} supports high-quality table retrieval.
Table question answering is another representative application, and we have already demonstrated the strong performance of \ours{} in retrieving question-relevant tables in Sec.~\ref{sec:topk_retrieval}.




In table-based fact verification, the goal is to determine whether a given factual claim is supported or refuted by a table.
Most prior work focuses on the verification step, assuming that the relevant table is already available~\cite{ChenWCZWLZW20,ChaiGZF021,GuF0NZ022}.
In contrast, \ours{} retrieves claim-relevant tables, which can be paired with existing verification models.
OpenTFV~\cite{OpenTFV22} adopts a similar two-stage approach: keyword matching retrieves the top-$k$ tables ($k=10$ in our setting) for high recall, followed by semantic re-ranking and final verification using models like LPA~\cite{ChenWCZWLZW20}.

We compare \ours{} with OpenTFV on TabFact using this retrieval-and-verification pipeline, where both methods employ the state-of-the-art PASTA model~\cite{GuF0NZ022} for fact verification to ensure a fair comparison.
As shown in Table~\ref{tab:TFV}, \ours{} achieves superior end-to-end verification performance compared to OpenTFV.
This improvement stems from more effective retrieval of relevant evidence tables, which directly enhances verification accuracy.
In addition, \ours{} achieves substantially lower total inference time across all queries, with up to two orders of magnitude speedup, as OpenTFV relies on TAPAS~\cite{HerzigNMPE20}, a large pre-trained model for embedding generation that incurs significant computational overhead.

\setlength{\textfloatsep}{0pt}
\begin{table}[t]
\setlength{\tabcolsep}{3pt}
\setlength{\abovecaptionskip}{0cm}
\setlength{\belowcaptionskip}{0cm}
\footnotesize
\centering
\caption{Impact of training set size (``Size'' indicates the percentage of the entire dataset used for training)}
\label{tab:size}
\begin{tabular}{ccccccccccc}
\toprule
\multirow{2}{*}{Dataset} & \multirow{2}{*}{Size} & \multicolumn{3}{c}{\ours} & \multicolumn{3}{c}{CMDL} & \multicolumn{3}{c}{PASTA} \\
\cmidrule(lr){3-5} \cmidrule(lr){6-8} \cmidrule(lr){9-11}
&&P&R&F1&P&R&F1&P&R&F1\\
\midrule
\multirow{4}{*}{UK-Open}&1\%&0.968&0.991&0.979&0.457&0.785&0.577&N/A&N/A&N/A\\
&5\%&0.979&0.994&0.986&0.466&0.799&0.588&N/A&N/A&N/A\\
&10\%&0.982&0.997&0.989&0.487&0.829&0.613&N/A&N/A&N/A\\
&20\%&0.987&0.998&0.992&0.507&0.852&0.636&N/A&N/A&N/A\\
\hline
\multirow{4}{*}{MMQA}&1\%&0.65&0.753&0.698&0.367&0.201&0.26&N/A&N/A&N/A\\
&5\%&0.653&0.761&0.703&0.422&0.213&0.283&N/A&N/A&N/A\\
&10\%&0.661&0.768&0.71&0.446&0.241&0.313&N/A&N/A&N/A\\
&20\%&0.663&0.772&0.713&0.463&0.255&0.329&N/A&N/A&N/A\\
\hline
\multirow{4}{*}{TabFact}&1\%&0.655&0.965&0.78&0.351&0.122&0.181&0.675&0.879&0.764\\
&5\%&0.659&0.979&0.788&0.374&0.131&0.194&0.699&0.881&0.779\\
&10\%&0.661&0.989&0.792&0.408&0.158&0.228&0.718&0.891&0.795\\
&20\%&0.665&0.994&0.797&0.428&0.163&0.236&0.731&0.899&0.806\\
\bottomrule
\end{tabular}
\vspace{-0.4cm}
\end{table}

\subsection{Sensitivity Analysis (Q4)}
\label{sec:analysis}

In this section, we analyze the sensitivity of \ours{} with respect to the training set size and the number of neighbor links ($K$) used in graph construction.

\smallskip
\noindent\textbf{Impact of training set size.}
As discussed in Sec.~\ref{sec:intro}, \ours{} is designed to achieve strong performance under limited supervision.
To evaluate this property, we examine how its performance varies with different training set sizes and compare it against relevant baselines under the same training conditions.

We select three datasets -- UK-Open, MMQA, and TabFact -- that involve diverse types of NL statements and were among the more challenging cases in our earlier experiments (Sec.~\ref{sec:NL2table}).
We focus on comparing \ours{} with CMDL across three datasets, as CMDL is trained using the same training splits and thus allows for a fair comparison under varying training sizes.
Although Llama4-Scout and MTR-TableLlama were included in earlier experiments on MMQA, we exclude them because they are trained or fine-tuned in prior work and are not adapted to our training data, making them unsuitable for this evaluation.
For TabFact, we additionally include PASTA, which is fine-tuned on our training data and serves as a strong baseline for table-based fact verification.

As shown in Table~\ref{tab:size}, \ours{} consistently outperforms CMDL across all three datasets in terms of precision, recall, and F1 score, even when trained on smaller subsets of the data.
On TabFact, \ours{} achieves competitive performance compared to PASTA.
While PASTA attains higher precision, \ours{} consistently yields substantially higher recall.
Notably, when trained on only 1\% or 5\% of the full dataset, \ours{} surpasses PASTA in F1 score, highlighting its advantage under limited supervision.
Overall, these results demonstrate the robustness of \ours{} in low-resource settings, where only a small number of labeled examples are available.

\setlength{\textfloatsep}{0pt}
\begin{table}[t]
\setlength{\tabcolsep}{3pt}
\setlength{\abovecaptionskip}{0cm}
\setlength{\belowcaptionskip}{0cm}
\footnotesize
\centering
\caption{Impact of the number of neighbor links $K$}
\label{tab:K}
\begin{tabular}{cccccccccc}
\toprule
\multirow{2}{*}{K} & \multicolumn{3}{c}{UK-Open} & \multicolumn{3}{c}{MMQA} & \multicolumn{3}{c}{TabFact} \\
\cmidrule(lr){2-4} \cmidrule(lr){5-7} \cmidrule(lr){8-10}
&P&R&F1&P&R&F1&P&R&F1\\
\midrule
1  & 0.984 & 0.996 & 0.990 & 0.660 & 0.767 & 0.710 & 0.661 & 0.992 & 0.794 \\
5  & 0.986 & 0.997 & 0.992 & 0.661 & 0.770 & 0.712 & 0.663 & 0.993 & 0.795 \\
10 & 0.987 & 0.998 & 0.992 & 0.663 & 0.772 & 0.713 & 0.665 & 0.994 & 0.797 \\
15 & 0.988 & 0.999 & 0.994 & 0.665 & 0.773 & 0.715 & 0.667 & 0.995 & 0.798 \\
20 & 0.989 & 0.999 & 0.994 & 0.666 & 0.775 & 0.716 & 0.668 & 0.996 & 0.799 \\
\bottomrule
\end{tabular}
\end{table}

\smallskip
\noindent\textbf{Impact of the number of neighbor links $K$}. We evaluate the sensitivity of \ours{} to the number of neighbor links $K$ using three datasets -- UK-Open, MMQA, and TabFact.
As shown in Table~\ref{tab:K}, increasing $K$ leads to modest improvements in precision, recall, and F1 score.
However, the overall impact is limited, and \ours{} maintains strong performance even when $K=1$.
This result suggests that while richer contextual connectivity can benefit data discovery, \ours{} remains robust to sparse graph construction.

\section{Related Work}
\label{sec:relatedWork}

\noindent\textbf{Data discovery in data lakes}. A substantial body of work has focused on data discovery in data lakes, primarily using either an exemplar table~\cite{ZhuDNM19,NargesianZPM18,Dong0NEO23,BogatuFP020,DongT0O21,FernandezAKYMS18} or keywords~\cite{ChenTHX020,BrickleyBN19,ZhangB18,ShragaRFC20} as input. They typically support a single query modality and assume a specific user intent, leveraging similarity signals such as value overlap~\cite{ZhuDNM19}, schema similarity~\cite{NargesianZPM18}, or distances between learned embeddings~\cite{DongT0O21,Dong0NEO23}. Systems such as Aurum~\cite{FernandezAKYMS18} and D3L~\cite{BogatuFP020} further combine multiple signals to improve retrieval accuracy.

More recent studies~\cite{EltabakhKEA23,wang2025towards,FengRFCK24,SilvaB23,LeesBKS0021} explore using unstructured text as queries, but they are still designed around specific user intents (e.g., open-ended text enrichment~\cite{SilvaB23,LeesBKS0021}). Among these, CMDL~\cite{EltabakhKEA23} is most closely related to our work, as it supports both table- and NL-based queries. However, CMDL assumes a targeted intent for each query type and does not generalize to diverse intents. In contrast, \ours{} accommodates diverse intents within a unified framework and consistently outperforms CMDL on both NL- and table-based discovery.

\smallskip
\noindent\textbf{Table question answering}. Table retrieval for question answering (QA) is a prominent user intent in data discovery, where the goal is to retrieve one or multiple tables required to answer a query. However, most existing work on table QA assumes that the relevant tables are already provided and focuses primarily on answer generation~\cite{HerzigNMPE20,YinNYR20,ZhangHFCDP24,LiuCGZLCL22}.
Only a few studies explicitly address table retrieval. Single-table retrieval methods such as Solo~\cite{WangF23} and Pneuma~\cite{balaka2025pneuma} aim to identify a single table that can answer a question. For multi-table QA, MTR-TableLLaMa~\cite{wummqa} decomposes complex questions into sub-questions and retrieves tables iteratively before answer generation. In contrast, \ours{} is designed as a unified table retrieval framework that directly supports this user intent. By jointly modeling relationships between questions and tables, \ours{} enables holistic retrieval without explicit question decomposition. Experimental results show that \ours{} consistently outperforms MTR-TableLLaMa and Pneuma in retrieving relevant tables, demonstrating its effectiveness in supporting table QA through robust data discovery.

\smallskip
\noindent Beyond the closely related work discussed above, we discuss additional related topics, including table-based fact verification~\cite{GuF0NZ022,HerzigNMPE20,ZhangWWCZW20,ZhouLZW22,ChaiGZF021,OpenTFV22}, NL-to-SQL~\cite{DLiJ14,SahaFSMMO16,GaoWLSQDZ24,LiHQYLLWQGHZ0LC23,TaiCZ0023} and heterogeneous graph neural networks~\cite{WuPCLZY21,WangJSWYCY19,fu2020magnn,ZhangSHSC19,wang2021self,HuDWS20}, in Appendix~\ref{ap:add_work}.

\section{Conclusion}
\label{sec:conclusion}
We introduced \ours{}, a unified data discovery framework for data lakes that supports both natural language statements and tables as queries, while generalizing across diverse user intents in discovery scenarios, such as open-ended enrichment and task-driven inference. 
At its core, \ours{} employs a cross-modal representation module that learns unified embeddings for NL statements and tables via a graph learning model under limited supervision, by exploiting naturally occurring contextual relationships in data lakes to substantially reduce reliance on large collections of labeled (query, table) pairs. 
These learned embeddings are materialized and indexed, enabling an adaptive retrieval module to efficiently support both binary and ranked retrieval across diverse discovery scenarios. 
Extensive experiments on seven real-world datasets demonstrate that \ours{} enables effective and efficient data discovery for both NL- and table-based queries, consistently outperforming strong baselines on downstream tasks such as table question answering and table-based fact verification, while remaining robust under limited supervision and sparse graph structures.


\clearpage


\bibliographystyle{ACM-Reference-Format}
\bibliography{totalab}

\clearpage

\appendix

\section{An Example for Joint Optimization Objective}
\label{ap:exam_loss}

This example shows how the joint learning objective in \ours{} combines contrastive learning loss and edge reconstruction loss.

\begin{example}
As illustrated in Fig.~\ref{fig:loss}, based on the heterogeneous graph with node embeddings obtained using meta-path-based neighbor aggregation, we sample three observed edges, $s_1$–$t_2$, $s_1$–$t_5$, and $t_3$–$t_4$, as masked positive edges, and three unconnected pairs, $s_1$–$t_1$, $s_1$–$t_3$, and $t_1$–$t_3$, as negative edges. 
These edges are decoded using an MLP model, and the MLP outputs are used to compute the edge reconstruction loss $\mathcal{L}_{\mathrm{ER}}$.

To compute the contrastive learning loss, we consider both meta-path-based neighbor and $r$-neighbor aggregation views. 
For target nodes $s_1$ and $t_1$, we treat embeddings under $r$-neighbor aggregation as positive samples, and use the embeddings of another node under $r$-neighbor aggregation -- $s_3$, $s_4$ for $s_1$, and $t_3$, $t_4$ for $t_1$ -- as negative samples to compute the loss $\hat{\mathcal{L}}_{\mathrm{CL}}^{\mathcal{P}}$. 
Similarly, we can compute the loss $\hat{\mathcal{L}}_{\mathrm{CL}}^{r}$. 
The total contrastive loss $\mathcal{L}_{\mathrm{CL}}$ is obtained by combining $\hat{\mathcal{L}}_{\mathrm{CL}}^{\mathcal{P}}$ with $ \hat{\mathcal{L}}_{\mathrm{CL}}^{r}$.

Finally, we combine the contrastive learning loss $\mathcal{L}_{\mathrm{CL}}$ and the edge reconstruction loss $\mathcal{L}_{\mathrm{ER}}$ to form the overall objective $\mathcal{L}$.
\end{example}

\begin{figure}[h]
\setlength{\abovecaptionskip}{0cm}
\setlength{\belowcaptionskip}{0cm}
\centering
\includegraphics[width=85mm]{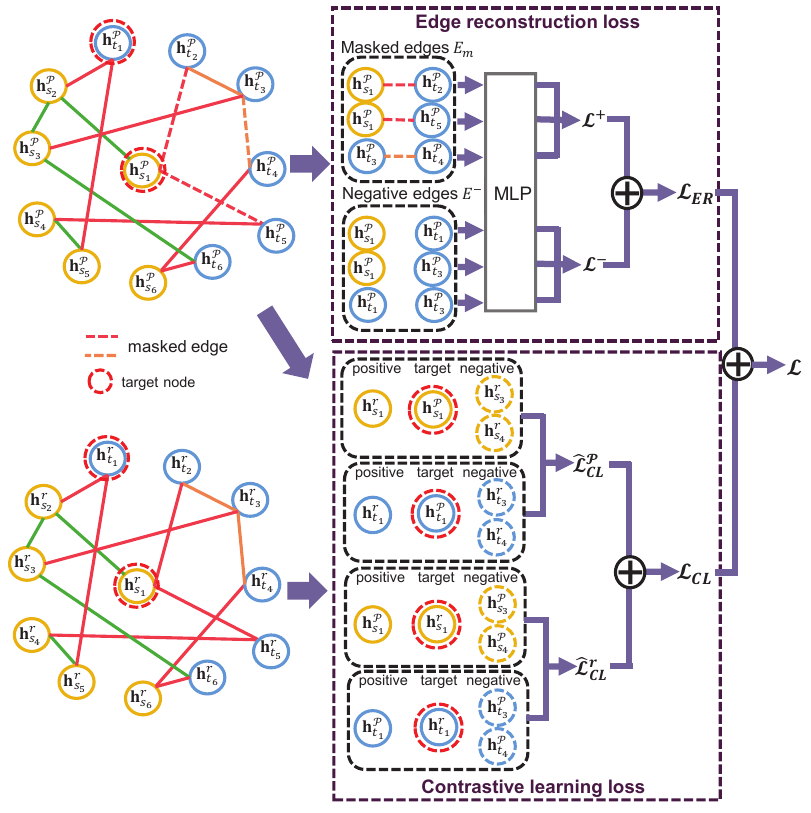}
\caption{\label{fig:loss}{The illustration of joint optimization objective.} 
}
\end{figure}

\section{Additional Related Work}
\label{ap:add_work}

\noindent\textbf{Table-based fact verification}. Table-based fact verification is an important user intent in data discovery, where the goal is to verify natural language claims using structured tables as evidence. However, most prior work assumes relevant tables are readily available and focuses solely on the verification step~\cite{GuF0NZ022,HerzigNMPE20,ZhangWWCZW20,ZhouLZW22,ChaiGZF021}. To bridge this gap, OpenTFV~\cite{OpenTFV22} introduces a two-stage pipeline: (i) keyword matching to retrieve the top-$k$ candidate tables, and (ii) semantic re-ranking and final verification using models such as LPA~\cite{ChenWCZWLZW20}. In our experiments, we adopt this pipeline but replace LPA with PASTA~\cite{GuF0NZ022}, a state-of-the-art verification model. Results show that \ours{} substantially improves retrieval quality, leading to better end-to-end verification performance.

\smallskip
\noindent\textbf{Natural Language to SQL (NL-to-SQL)}. NL-to-SQL approaches
aim to translate natural language queries into executable SQL
statements. Early methods relied on syntactic parsers and handcrafted rules~\cite{DLiJ14,SahaFSMMO16}, followed by pre-trained language models
such as BERT~\cite{DevlinCLT19} and T5~\cite{RaffelSRLNMZLL20}. More recent works~\cite{GaoWLSQDZ24,LiHQYLLWQGHZ0LC23,TaiCZ0023} leverage large language models (LLMs) with advanced prompting
techniques for SQL generation. A key challenge in these methods
is schema linking – accurately aligning natural language phrases
with the corresponding table and column names. This is typically
addressed under the assumption that the relevant tables are provided with complete schemas. Moreover, these methods primarily
focus on query answering -- generating effective SQL queries based
on given tables. In contrast, our work tackles the problem of data
discovery: identifying relevant tables from a large corpus, often
with incomplete or missing schema information. This setting is
more realistic in large-scale data lakes and serves as a crucial prerequisite for downstream query answering.

\smallskip
\noindent\textbf{Heterogeneous Graph Neural Networks}. Graph neural networks (GNNs) have attracted considerable attention in recent years. 
Most GNNs are homogeneous, as summarized in the comprehensive survey by~\cite{WuPCLZY21}. 
Recently, increasing efforts have been made to extend GNNs to heterogeneous data. 
For example, HAN~\cite{WangJSWYCY19} introduces both node-level and semantic-level attention mechanisms to learn node embeddings.
MAGNN~\cite{fu2020magnn} incorporates intra-metapath and inter-metapath aggregation to enhance representation learning.
HetGNN~\cite{ZhangSHSC19} jointly considers heterogeneous content encoding, type-based neighbor aggregation, and the integration of different node types to enhance representation learning.
HeCo~\cite{wang2021self} employs a cross-view contrastive learning framework to capture both local and high-order structural information. 
HGT~\cite{HuDWS20} is designed to handle web-scale heterogeneous graphs effectively. 
Inspired by these recent advances, \ours{} introduces a specialized heterogeneous graph neural network tailored to effectively learn embeddings for both NL statements and tables.

\end{document}